\RequirePackage{fix-cm}
\documentclass{svjour3}                     
\smartqed  
\usepackage{graphicx}
\usepackage{mathptmx}      
\usepackage{latexsym}
\usepackage{amsmath,epsfig}
\usepackage{amssymb}
\usepackage[utf8]{inputenc}  
\usepackage[T1]{fontenc} 
\usepackage{enumitem}
\usepackage{mathtools}
\usepackage{amssymb,amsfonts}
\usepackage{dsfont}
\usepackage{dsfont}
\usepackage{stmaryrd}
\usepackage{bm}
\usepackage{color}
\usepackage{soul}

\def\bx{{\bf x}}
\def\bW{{\bf W}}
\def\bX{{\bf X}}
\def\bY{{\bf Y}}
\def\bZ{{\bf Z}}
\def\bZt{\widetilde{\bf Z}}
\begin{document}

\title{The adaptive interpolation method\\ for proving replica formulas. Applications to \\the Curie-Weiss and Wigner spike models}

\titlerunning{The adaptive interpolation method for proving replica formulas.}        

\author{Jean Barbier         \and
        Nicolas Macris 
}


\institute{J. Barbier$^{\dagger\star}$ \& N. Macris$^{*}$ \at
              $\dagger$ Quantitative Life Sciences, The Abdus Salam International Center for Theoretical Physics, Trieste, Italy.\\      
              $\star$ Physics Laboratory, Ecole Normale Sup\'erieure, Paris, France.\\
              $*$ Communication Theory Laboratory, Ecole Polytechnique F\'ed\'erale de Lausanne, Switzerland.\\          
              \email{jbarbier@ictp.it, nicolas.macris@epfl.ch}            
}

\maketitle

\begin{abstract}
In this contribution we give a pedagogic introduction to the newly introduced {\it adaptive interpolation method} to prove in a simple and unified way replica formulas for Bayesian optimal 
inference problems. Many aspects of this method can already be explained at the level of the simple Curie-Weiss spin system. This provides a new method of solution for this model
which does not appear to be known. We then generalize this analysis to a paradigmatic inference problem, namely rank-one matrix estimation, also refered to as the Wigner spike model in statistics. We give many pointers to the recent literature
where the method has been succesfully applied.
\end{abstract}

\section{Introduction}
\label{intro}
The replica method from statistical mechanics has been applied to Bayesian inference problems (e.g., coding, estimation) already two decades ago \cite{Tanaka,kabashima2003cdma}. Rigorous proofs 
of the formulas for the mutual informations/entropies/free energies stemming from this method, have for a long time only been partial, consisting generally of one sided 
bounds \cite{franz2003replica,Franz-Leone-Toninelli,KoradaMacris_CDMA,MontanariTse,Montanari-codes,Macris05CorrInequ,Macris07LDPC,Kudekar-Macris-2009}. It is only quite recently that there has been a surge of progress 
using various methods --namely spatial coupling \cite{Giurgiu_SCproof,XXT,2018arXiv181202537B,barbier_ieee_replicaCS,barbier_allerton_RLE}, 
information theory \cite{private}, and rigorous versions of the cavity method \cite{coja-2016,2016arXiv161103888L,2017arXiv170108010L}-- 
to derive full proofs, but which are typically quite complicated. 
Recently we 
introduced \cite{BarbierM17a} a powerful evolution of the Guerra-Toninelli \cite{generalised-mean-field,guerra2002thermodynamic,guerra2005introduction} 
interpolation method --called {\it adaptive interpolation}-- that allows to fully prove the replica formulas in a quite simple and unified way for Bayesian inference problems. In this contribution we give a pedagogic introduction to this new method. 

We first illustrate how the adaptive interpolation method allows to solve the 
well known Curie-Weiss model. For this model a simple version of the method contains most of the crucial ingredients. The 
present rigorous method of solution is new and perhaps simpler compared to the usual ones.  
Most of these ideas can be transferred to any Bayesian inference problem, and here this is reviewed in detail for the {\it rank-one matrix estimation} or factorisation problem (one of the simplest non-linear estimation problems). The solution of Bayesian inference problems requires only one supplementary ingredient:
 the concentration 
of the ``overlap'' with respect to {\it both} thermal and quenched disorder. Remarkably, this can be proven for Bayesian inference problems in a setting often called {\it Bayesian optimal inference}, i.e., when the prior and hyper-parameters are all known.

The adaptive interpolation method has been fruitfuly applied to a range of more difficult problems with a dense 
underlying graphical structure. So far these include 
matrix and tensor factorisation \cite{2017arXiv170910368B}, estimation in traditional and generalised 
linear models \cite{barbier2017phase} (e.g., compressed sensing and many of its non-linear variants), with random i.i.d. as well as special structured measurement matrices \cite{RLEStructuredMatrices}, learning 
problems in the teacher-student setting \cite{barbier2017phase,NIPS2018_7584} (e.g., the single-layer perceptron network) and even multi-layer versions \cite{NIPS2018_7453}. 
For inference problems with an underlying sparse graphical structure full proofs of replica formulas are scarce and much more involved, e.g., \cite{Giurgiu_SCproof,coja-2016}. The interpolation method and replica bounds for sparse systems have been pioneered by Franz and Leone in \cite{franz2003replica,Franz-Leone-Toninelli} (see also \cite{bayati2013,salez2016interpolation}) but so far the present adaptive interpolation method is still in its infancy for sparse systems \cite{eric-censor-block} and it would be desirable to develop it further.
The method was initially formulated with a more technical discrete interpolation scheme, but it was already observed that a continuous 
interpolation is natural \cite{BarbierM17a}. The continuous analysis presented here was then explicitly developed for the tensor factorization problem in \cite{2017arXiv170910368B}. Here we directly use the continuous version which is certainly more natural when the underlying graphical structure is dense. So far however, for sparse systems which present new difficulties, only the discrete analysis has been developed \cite{eric-censor-block}.

\section{The Curie-Weiss model revisited}
\label{sec:format}
Consider the Hamiltonian of the ferromagnetic Ising model on a complete graph of spins $\boldsymbol{\sigma}\in\{-1,1\}^n$, or Curie-Weiss model:
\begin{align*}
{\cal H}(\boldsymbol{\sigma}) \equiv -\frac{J}{n} \sum_{i<j} \sigma_i\sigma_j - h\sum_{i=1}^n\sigma_i
\end{align*}
where $J>0$ is the coupling constant and $h\in\mathbb{R}$ the external field. 
The free energy 
is 
\begin{align*}
f_n^{\rm cw} \equiv -  \frac{1}{\beta n} \ln \mathcal{Z}_n  = - \frac{1}{\beta n} \ln\sum_{\boldsymbol{\sigma}\in\{-1,1\}^n} e^{-\beta {\cal H}(\boldsymbol{\sigma})}
\end{align*}
with $\mathcal{Z}_n$ the partition function and $\beta>0$ the inverse temperature. 
%
%

Let ${\cal A}\equiv[\beta(h-J),\beta(h+J)]$ and $f^{\rm pot}(m,\widehat m)=f^{\rm pot}(m,\widehat m; h,J,\beta)$ the {\it potential}:
\begin{equation}\label{pot_2}
f^{\rm pot}(m,\widehat m) \equiv - \frac1\beta \ln(2\,{\rm cosh}\,\widehat m) - \frac{Jm^2}{2} + \Big(\frac{\widehat m }\beta-h\Big)m\,. 
\end{equation}
Note that this potential verifies at its stationary point(s)
\begin{align}\label{6_}
\begin{cases}
\partial_{m} f^{\rm pot}(m,\widehat m) = 0 & \Leftrightarrow \quad\widehat m = \beta (Jm + h)\,,\\
\partial_{\widehat m} f^{\rm pot}(m,\widehat m) = 0 &\Leftrightarrow\quad  m = {\rm tanh}\,\widehat m\,.
\end{cases}
\end{align}
The purpose of this section is to prove in a new fashion the following variational formula for the free energy:
\begin{theorem}\label{thmCW}
The thermodynamic limit of the free energy for the Curie-Weiss model verifies
\begin{align*}
\lim_{n\to\infty} f_n^{\rm cw} &=
 {\adjustlimits \inf_{\widehat m\in{\cal A}}\sup_{m\in [-1, 1]}} f^{\rm pot}(m,\widehat m) \,.
\end{align*}	
\end{theorem}
\begin{remark}\label{simpleexpressionCW}
A trivial maximization over $m$ yields
\begin{align*}
{\adjustlimits \inf_{\widehat m\in{\cal A}}\sup_{m\in [-1, 1]}} f^{\rm pot}(m,\widehat m) & = \inf_{\widehat m\in{\cal A}} \Big\{\frac{1}{2J}
\Big(\frac{\widehat m}{\beta} -h\Big)^2-\frac{1}{\beta}\ln(2\,{\rm cosh}\,\widehat m) \Big\}
\\ &
=
\inf_{m\in [-1, 1]} \Big\{\frac{J m^2}{2}-\frac{1}{\beta}\ln\big(2\,{\rm cosh}\,[\beta\{Jm+h\}]\big)\Big\} 
\end{align*}
where the second equality is obtained from the first one by parametrizing $\widehat m = \beta(Jm + h)$, $m\in [-1, 1]$. 
This well known variational formula can also be obtained directly by a simpler formulation of the adaptive interpolation discussed 
in section \ref{anotherpath_cw}. However, this simpler formulation is not powerful enough for more complicated problems.
\end{remark}
\begin{remark}
It is possible to check that 
$$
{\adjustlimits \inf_{\widehat m\in{\cal A}}\sup_{m\in [-1, 1]}} f^{\rm pot}(m,\widehat m)= {\adjustlimits \inf_{m\in [-1, 1]}\sup_{\widehat m\in{\cal A}}}\, f^{\rm pot}(m,\widehat m)\, .
$$
For the present model this can be checked by explicit computation of both sides. Otherwise, this follows 
directly using that the potential equals $\widehat m\, m/\beta$ minus a convex function of $m$ minus another convex function of $\widehat m$ (see e.g., the appendix D of \cite{barbier2017phase}). Finally we also note that maximization over $\widehat m$ in the r.h.s. yields another well known formula
for the free energy:
$$
\lim_{n\to\infty} f_n^{\rm cw} = 
\inf_{m\in [-1, +1]} \Big\{-\frac{J m^2}{2} - hm - \frac1\beta h_2(m) \Big\}
$$
with $h_2(m) \equiv - \frac{1+m}{2} \ln\frac{1+m}{2} - \frac{1-m}{2} \ln \frac{1-m}{2}$ the binary entropy function.
\end{remark}

\subsection{Adaptive interpolation}
Before entering the proof let us give the generic roadmap of the adaptive interpolation method, and emphasize the main differences with the canonical Guerra-Toninelli interpolation method \cite{generalised-mean-field,guerra2002thermodynamic,guerra2005introduction}. 

The aim of the method is to prove a variational formula (such as Theorem~\ref{thmCW}) for the thermodynamic limit of the free energy of some complex statistical model of interacting variables/spins. This variational formula corresponds to the extremization of a proper potential (in the present case \eqref{pot_2}). $i)$ The first step consists in defining an ``interpolating model'' parametrized by ``time'' $t\in[0,1]$. Its associated $t$-dependent free energy $f(t)$ must interpolate between the one of the model of interest at, say, $t=0$, and the one of a properly chosen ``decoupled model'' at $t=1$ where the variables do not interact anymore and with a tractable free energy (tractable because the system is decoupled) which constructs part of the potential. The basic idea is therefore similar as the canonical interpolation method except that usually the interpolation path depends ``trivially'' on $t$ while in the adaptive interpolation method the interpolation path depends on $t$ in a totally generic way through the definition of {\it interpolation functions} that allow for much more flexibility. $ii)$ In the second step we want to ``compare'' the two boundary values of the free energy using $f(0)=f(1)-\int_0^1 dt\,f'(t)$ as usual, where $f(0) = f$ is what we want to compute while $f(1)$ is a piece of the potential. We therefore need to compute the $t$-derivative $f'(t)$. When $f'(t)$ is then plugged in the previous relation this gives the so-called sum rule, which links the free energy of interest and the potential (or part of it). $iii)$ The third step consists in simplifying the obtained sum rule thanks to the concentration of the identified order parameter of the problem (the magnetization in the Curie-Weiss model, or the overlap with the planted solution/signal in Bayesian inference problems). Self-averaging of the order parameter, often refered to as {\it replica symmetry} \cite{MezardParisi87b}, has to be proven for all $t\in[0,1]$. It requires a proper ``perturbation'' of the model with a strenght controlled by a perturbation parameter $\epsilon$. Perturbing the system allows to ``avoid'' possible isolated phase transitions points where concentration does not occur. This step is model-dependent as such results can be proven only under specific settings. For ferromagnetic models (such as the Curie-Weiss model) at any temperature, Bayesian inference (in the so-called ``Bayesian optimal setting'' discussed in sec.~\ref{sec:matrix}), or generic disordered spin models at high temperature this is doable. In the first case thanks to the ferromagnetic nature of the model (see \cite{2019arXiv190106521B} for a proof that ferromagnetism implies replica symmetry in full generality), in the second case thanks to special identities called ``Nishimori identities'' (see appendix (A1)) for the correlation functions of the model, and in the last case by concentration techniques \cite{Talagrand2011spina,Talagrand2011spinb}. But away from these settings, e.g., in combinatorial optimization, in generic disordered spin models at low temperature, or in non-optimal Bayesian inference, this is usually not possible as {\it replica symmetry breaking} may occur \cite{MezardParisi87b,Talagrand2011spina,Talagrand2011spinb} and prevents the order parameter to concentrate\footnote{At the moment the adaptive interpolation method is specifically designed for replica symmetric models but it is an interesting direction to see wether it can tackle more complicated models, e.g., by being combined with ideas coming from \cite{Guerra-2003}.}. $iv)$ In a fourth step, once the sum rule has been simplified thanks to the order parameter concentration, the flexibility allowed by the choice of the interpolation functions (i.e., the choice of the interpolation path) is exploited in order to obtain two matching bounds between the potential and the free energy, and thus the result. One bound is simply obtained by choosing a ``trivial'' interpolation path. The other one requires a smarter choice: it appears that, given the decoupled model towards which we interpolate, there is a unique choice of the interpolation functions allowing to obtain the converse bound. This choice corresponds to the solution of a first order differential equation over the interpolation functions (in which the perturbation parameter $\epsilon$ will play the role of the initial condition). The interpolation functions have therefore been {\it adapted} in order to finish the proof, thus the name of the method.\\

We now formalize these ideas. Let $\epsilon \in[s_n, 2s_n]$ for a sequence $s_n\in(0,1/2]$ that tends to $0_+$ as $(1/2)n^{-\alpha}$ for $0<\alpha<1$. 
 Here $\epsilon$ is interpreted as a ``perturbation field'' that will soon play a crucial role (note that this field could also belong to $[-s_n, s_n]$ without changing the sub-sequent analysis) . Let $\widehat m_{\epsilon}: [0, 1] \mapsto  {\cal A}$ a ``trial'' {\it interpolation function} depending on an interpolating parameter $t\in [0, 1]$ and on $\epsilon$. The interpolation function is at the moment generic, and will be chosen (adapted) later on. Then set $\widehat R_\epsilon(t)\equiv \epsilon+\int_0^tds\,\widehat m_\epsilon(s)$.
Define an interpolating Hamiltonian, with interpolating parameter 
$t\in [0, 1]$, as
\begin{equation*}
{\cal H}_{t,\epsilon}(\boldsymbol{\sigma}) \equiv  (1-t) {\cal H}(\boldsymbol{\sigma}) - \frac{\widehat R_\epsilon(t)}\beta \sum_{i=1}^n\sigma_i
\end{equation*}
and the corresponding interpolating free energy
\begin{align*}
 f_{n,\epsilon}(t) \equiv - \frac1{\beta n}\ln {\cal Z}_{n,\epsilon}(t) =  - \frac1{\beta n}\ln \sum_{\boldsymbol{\sigma}\in\{-1,1\}^n} e^{-\beta {\cal H}_{t,\epsilon}(\boldsymbol{\sigma})} \,.
\end{align*}
In this definition ${\cal Z}_{n,\epsilon}(t)$ is the partition function associated with the interpolating Hamiltonian. This model interpolates between the Curie-Weiss model at $t=0$ (with a slightly different external field $h+\epsilon/\beta$) and a decoupled spin model at $t=1$ with mean-field controlled by the trial function and the perturbation. It is then easy to verify that 
\begin{align}\label{boundaries}
\begin{cases}
f_{n,\epsilon}(0)  =f_{n}^{\rm cw}+O(s_n)\,,\\
f_{n,\epsilon}(1)  = - \beta^{-1}\ln(2 \,{\rm cosh}\,\widehat R_\epsilon(1))=- \beta^{-1}\ln(2 \,{\rm cosh}\,\int_0^1dt\, \widehat m_\epsilon(t))+O(s_n)\,.
\end{cases} 
\end{align}
To obtain the equality in the first line we use that $|df_{n}^{\rm cw}/dh| = |\langle M\rangle|\le1$, where the magnetization $M(\boldsymbol{\sigma})= M \equiv \frac{1}{n}\sum_{i=1}^n \sigma_i\in[-1,1]$ and $\langle - \rangle$ is the thermal average, i.e., the expectation w.r.t. the measure proportional to $e^{-\beta {\cal H}(\boldsymbol{\sigma})}$. Therefore $\vert f_{n}^{\rm cw}(h+\epsilon/\beta) - f_n^{\rm cw}(h)\vert \leq 2s_n/\beta$ by the mean value theorem, where $f_n^{\rm cw}(h)=f_n^{\rm cw}$. 
In the second equality of the second line the perturbation term $\epsilon\in[s_n,2s_n]$ has been extracted by continuity. 

In order to compare the free energy of the Curie-Weiss model with the potential we use the fundamental theorem of calculus. Using \eqref{boundaries} we find directly
\begin{align}
f_{n}^{\rm cw} = f_{n,\epsilon}(1) - \int_0^1 dt \frac{df_{n,\epsilon}(t)}{dt} + O(s_n)\, . \label{thCalc}
\end{align}
We now naturally turn to the calculation of $df_{n,\epsilon}/dt$. Let us introduce the notation $\langle -\rangle_{t,\epsilon}$ for the Gibbs bracket 
of the interpolating system, which is the thermal average of a function $A=A(\boldsymbol{\sigma})$ of the spins:
\begin{equation*}
\langle A\rangle_{t,\epsilon} \equiv \frac{1}{\mathcal{Z}_{n,\epsilon}(t)} \sum_{\boldsymbol{\sigma}\in{\{-1, 1\}}^n}e^{-\beta \mathcal{H}_{t,\epsilon}(\boldsymbol{\sigma})}\,A(\boldsymbol{\sigma})\,.
\end{equation*}
We compute 
\begin{align}\label{tder_ising}
\frac{df_{n,\epsilon}(t)}{dt} &=  \frac{1}{n}\Big\langle \frac{d{\cal H}_{t,\epsilon}}{dt}\Big\rangle_{t,\epsilon} = \frac{J}{2} \langle M^2 \rangle_{t,\epsilon}  +\Big(h-\frac{\widehat m_\epsilon(t)}\beta \Big)\langle M \rangle_{t,\epsilon}+ O(1/n)\,.
\end{align}
Replacing \eqref{boundaries} and \eqref{tder_ising} in \eqref{thCalc} yields the following fundamental {\it sum rule}:
\begin{align*}
f_n^{\rm cw} 
=  &- \frac1\beta\ln\Big(2 \,{\rm cosh}\int_0^1dt \,\widehat m_\epsilon(t)\Big)- \int_0^1dt \Big\{ \frac{J}{2}\langle M^2 \rangle_{t,\epsilon} +\Big(h-\frac{\widehat m_\epsilon(t)}\beta \Big)\langle M \rangle_{t,\epsilon}\Big\} + O(s_n)\,.
\end{align*}
\subsection{Simplifying the sum rule: concentration of the magnetization} \label{sec:2.2}
At this stage we need a concentration result for the magnetization, that states
\begin{align}\label{concen-cw}
\frac{1}{s_n}\int_{s_n}^{2s_n} d\epsilon\, \big\langle (M - \langle M\rangle_{t, \epsilon})^2\big\rangle_{t,\epsilon} \le  \frac2{ns_n}
\end{align}
and which holds for all values of the temperature, coupling constant and magnetic field. This is where the perturbation field $\epsilon$ plays a crucial role. 
For the Curie-Weiss model the proof is elementary and goes as follows. The thermal fluctuations of the magnetization are precisely given by the second derivative of the interpolating free energy (denoted $f_{n,\epsilon}(t,\widehat R_\epsilon(t))$ when we need to emphasize its explicit $\widehat R_\epsilon$-dependence) w.r.t. $\widehat R_\epsilon\equiv \widehat R_\epsilon(t)$:
\begin{align}
\frac{d^2 f_{n,\epsilon}}{d\widehat R_\epsilon^2}(t,\widehat R_\epsilon(t))  = - \frac{n}{\beta} \big\langle (M - \langle M\rangle_{t, \epsilon})^2\big\rangle_{t,\epsilon}\,.\label{2ndDer}
\end{align} 
Assume that the map $\epsilon \mapsto \widehat R_\epsilon(t)$ is a ${\cal C}^1$ diffeomorphism\footnote{Recall a diffeomorphism is a bijection which is continuously differentiable and whose inverse is also continuously differentiable.} whose Jacobian $\partial_\epsilon \widehat R_\epsilon(t)$ is greater or equal to one for all $t\in[0,1]$; we will say in this case that $\epsilon \mapsto \widehat R_\epsilon(t)$ is {\it regular}. Under this assumption
we can write
\begin{align*}
&\int_{s_n}^{2s_n} d\epsilon\,   \big\langle (M - \langle M\rangle_{t, \epsilon})^2\big\rangle_{t,\epsilon} \le \int_{\widehat R_{s_n}(t)}^{\widehat R_{2s_n}(t)} d\widehat R_\epsilon\,   \big\langle (M - \langle M\rangle_{t, \epsilon})^2\big\rangle_{t,\epsilon} \,.
\end{align*}
Then using \eqref{2ndDer} this leads to
\begin{align*}
\frac1{s_n}\int_{s_n}^{2s_n} d\epsilon\,  & \big\langle (M - \langle M\rangle_{t, \epsilon})^2\big\rangle_{t,\epsilon} \le - \frac{\beta}{ns_n} \, \Big(\frac{d f_{n, \epsilon}}{d\widehat R_\epsilon}(t,\widehat R_{2s_n}(t)) - \frac{d f_{n, \epsilon}}{d\widehat R_\epsilon}(t,\widehat R_{s_n}(t))\Big)\, .
\end{align*}
Moreover note that $d f_{n, \epsilon}/ d\widehat R_\epsilon= - \langle M\rangle_{t,\epsilon}/\beta$ which is bounded by $1/\beta$ in absolute value. 
Therefore the r.h.s. is bounded by $2/(ns_n)$ which proves \eqref{concen-cw}.


Now, integrating the fundamental sum rule over $\epsilon$ and using this concentration result (with a sequence $s_n$ vanishing more slowly than $1/(ns_n)$ as $n\to+\infty$, i.e., with $0<\alpha <1/2$)
\begin{align}
f_n^{\rm cw}
=  \frac{1}{s_n}\int_{s_n}^{2s_n} d\epsilon \Big[&-\frac1\beta\ln\Big(2 \,{\rm cosh}\int_0^1dt\, \widehat m_\epsilon(t)\Big)\nonumber\\
&- \int_0^1 dt\Big\{ \frac{J}{2}\langle M \rangle_{t,\epsilon}^2 +\Big(h-\frac{\widehat m_\epsilon(t)}\beta \Big)\langle M \rangle_{t,\epsilon}\Big\} \Big]+ O(s_n)
\label{sumrule}
\end{align}
where $O(s_n)$ is uniform in $t$, $\epsilon$ and $\widehat R_\epsilon$. Note that this identity is valid for an arbitrary trial function $\widehat m_\epsilon$ as long as $\epsilon \mapsto \widehat R_\epsilon(t)$ is regular, a condition that we have to verify when using this sum rule for specific trial functions.
\subsection{Matching bounds}
%
\subsubsection{Upper bound}
Fix the constant function $\widehat m_{\epsilon}(t)=\widehat m\in{\cal A}$. This trivially makes the map $\epsilon\mapsto \widehat R_{\epsilon}(t)=\epsilon+\widehat m \,t$ regular. Using this choice in the sum rule \eqref{sumrule}, and recalling the definition \eqref{pot_2} of the potential, we directly obtain
\begin{align*}
f_{n}^{\rm cw} = \frac{1}{s_n}\int_{s_n}^{2s_n}d\epsilon\int_0^1dt\, f^{\rm pot}(\langle M \rangle_{t,\epsilon},\widehat m) +O(s_n)\le \sup_{m\in[-1,1]} f^{\rm pot}(m,\widehat m)+O(s_n)\,.
\end{align*}
This is is true for any $\widehat m\in{\cal A}$, therefore $\limsup_{n\to\infty}f_{n}^{\rm cw}(h)\le {\adjustlimits\inf_{\widehat m\in{\cal A}}\sup_{m\in[-1,1]}} f^{\rm pot}(m,\widehat m)$.
\subsubsection{Lower bound}
Now we choose $\widehat m_\epsilon(t)$ to be the solution of
\begin{align}\label{diff-equ}
\widehat m_\epsilon(t) &= \beta(J\langle M \rangle_{t,\epsilon} + h) \,.
\end{align}
Here one has to be careful and ask whether this equation possesses a solution, as the r.h.s. depends on the interpolation path through the function $\widehat R_\epsilon(t)=\epsilon+\int_0^tds\, \widehat m_\epsilon(s)$. Here is a crucial observation: from the set-up of the interpolation, the l.h.s. 
is $\widehat m_\epsilon(t) = d\widehat R_\epsilon(t)/dt$ and the r.h.s. 
 is a function $F_n: [0, 1]\times [s_n-\beta(|h|+J), 2s_n + \beta(|h|+J)] \mapsto 
F_n(t , \widehat R_\epsilon(t))=\beta(J \langle M\rangle_{t,\epsilon} + h) \in{\cal A}$ so \eqref{diff-equ} is a first order differential equation (ODE)
\begin{align}\label{crucial}
\frac{d \widehat R_\epsilon(t)}{dt} = F_n(t , \widehat R_\epsilon(t))\quad\text{with initial condition}\quad \widehat R_\epsilon(0)=\epsilon\,.
\end{align}
Thus the ``perturbation'' actually serves as initial condition of this ODE. 
The function $F_n$ is ${\cal C}^1$ with bounded derivative w.r.t. its second argument. Indeed $\partial_{\widehat R_\epsilon} F_n(s , \widehat R_\epsilon(s))= \beta Jn\langle (M -\langle M\rangle_{s,\epsilon})^2\rangle_{s,\epsilon}$ which is finite for finite $n$. 
Therefore we can apply the Cauchy-Lipshitz theorem to assert that \eqref{crucial} possesses a unique global solution over $[0, 1]$ that we denote $\widehat R_{\epsilon,n}^*(t)=\epsilon+\int_0^tds\,\widehat m_{\epsilon,n}^*(s)$, with $\widehat m_{\epsilon,n}^*:[0, 1] \mapsto {\cal A}$.

We want to replace this solution in \eqref{sumrule}. Thus we have to check that the flow $\epsilon \mapsto \widehat R_{\epsilon,n}^*(t)$ of this ODE is regular (i.e., a ${\cal C}^1$ diffeomorphism with Jacobian greater or equal to one). The argument is as follows. The flow is injective by unicity of the solution and ${\cal C}^1$ since $F_n$ is itself ${\cal C}^1$. By the Liouville formula for the Jacobian (see \cite{hartmanordinary} Corollary 3.1 in Chapter V)
\begin{align*}
\frac{\partial \widehat R_{\epsilon,n}^*(t)}{\partial \epsilon}=\exp\Big(\int_0^tds \frac{\partial F_n}{\partial \widehat R_\epsilon}(s,\widehat R_{\epsilon,n}^*(s))\Big) 
\end{align*}
which is greater than one because $\partial_{\widehat R_\epsilon} F_n(s , \widehat R_\epsilon(s))= \beta Jn\langle (M -\langle M\rangle_{s,\epsilon})^2\rangle_{s,\epsilon} \ge 0$. Also, this Jacobian never vanishes so the local inversion theorem combined with the fact that the flow is injective implies that this flow is a ${\cal C}^1$ diffeomorphism. We can thus use the sum rule \eqref{sumrule}.

The concavity of $-\ln(2\,{\rm cosh}\,x)$ allows to extract the $t$-integral from the first term in \eqref{sumrule}:
\begin{align}
f_n^{\rm cw} \!
&\ge \!\frac{1}{s_n}\int_{s_n}^{2s_n}d\epsilon \int_0^1dt  \Big\{\!
- \! \frac1\beta\ln\big(2 \,{\rm cosh}\, \widehat m_{\epsilon,n}^*(t)\big) \! - \! \frac{J}{2}\langle M \rangle_{t,\epsilon}^2 
\! + \!
\Big(\frac{\widehat m_{\epsilon,n}^*(t)}\beta-h \Big)\langle M \rangle_{t,\epsilon}\Big\} \!+\! O( s_n )\nonumber\\
&=\frac{1}{s_n}\int_{s_n}^{2s_n}d\epsilon\int_0^1 dt\, f^{\rm pot}(\langle M \rangle_{t,\epsilon},\widehat m_{\epsilon,n}^*(t))+ O(s_n)\label{sr_final}
\end{align}
by recognizing the expression of the potential \eqref{pot_2}. A crucial observation is that with our particular choice of interpolating function we have
\begin{align}
f^{\rm pot}(\langle M \rangle_{t,\epsilon},\widehat m_{\epsilon,n}^*(t)) = \sup_{m\in [-1,1]} f^{\rm pot}(m,\widehat m_{\epsilon,n}^*(t))\,.\label{supCond}
\end{align}
Indeed $m\mapsto f^{\rm pot}(m,\widehat m)$ is concave as $\partial^2f^{\rm pot}(m,\widehat m)/\partial m^2= -J\le 0$. Then, by \eqref{6_}, we have that $f^{\rm pot}(m,\widehat m)$ attains its maximum in $m$ precisely when $\widehat m=\beta(Jm +h)$, which implies \eqref{supCond} using that $\widehat m_{\epsilon,n}^*(t)$ solves \eqref{diff-equ}. Therefore \eqref{sr_final} becomes
\begin{align*}
f_n^{\rm cw} 
&\ge \frac{1}{s_n}\int_{s_n}^{2s_n} d\epsilon\int_0^1dt \sup_{m\in [-1,1]} f^{\rm pot}(m,\widehat m_{\epsilon,n}^*(t)) + O(s_n)\ge {\adjustlimits\inf_{\widehat m\in{\cal A}}\sup_{m\in [-1,1]}} f^{\rm pot}(m,\widehat m)+ O(s_n)
\end{align*}
and thus $\liminf_{n\to\infty}f_{n}^{\rm cw}\ge {\adjustlimits\inf_{\widehat m\in{\cal A}}\sup_{m\in [-1,1]}} f^{\rm pot}(m,\widehat m)$. This ends the proof of Theorem \ref{thmCW}. 
\section{Alternative formulation for the Curie-Weiss model in a random field}\label{anotherpath_cw}
In this section we repeat the proof but directly obtain the simpler expression of remark \ref{simpleexpressionCW} for the free energy, 
with the variational formula involving a potential depending on the single parameter $m$ representing the magnetization. 
Moreover we consider this time random i.i.d. external local fields $h_i\sim P_h$ in order to show that this is easily 
implemented in our approach (this case could have been considered also with the previous method); this model has been 
first rigorously treated in \cite{PhysRevB.15.1519}. For convenience we consider $P_h$ to be supported on $[-S,S]$. Standard 
limiting arguments allow to extend the support to the whole real line as long as the first few moments exist. 

Looking at the derivation below, the reader might wonder why we took a seemingly more complicated path in the previous section by introducing a two-parameter potential depending on both $m$ and an ``effective field'' $\widehat m$. This is because the two different proofs are, as it will soon become clear, based on different types of convexity arguments, and in many problems of interest such as generalized linear estimation \cite{barbier2017phase} or non-symmetric tensor estimation \cite{2017arXiv170910368B}, only the interpolation based on a two-parameter potential seems to be effective in order to obtain a full proof of replica formulas (instead of single-sided bounds reachable using a single-parameter potential). The arguments are very similar than in the previous section and we will be brief.

The Hamiltonian with random external fields, free energy and potential are this time
\begin{align*}
{\cal H}(\boldsymbol{\sigma};\textbf{h}) &\equiv -\frac{J}{n} \sum_{i<j} \sigma_i\sigma_j - \sum_{i=1}^n h_i\sigma_i\,,\\
f_n^{\rm rcw} &\equiv - \frac{1}{\beta n} \mathbb{E}\ln\sum_{\boldsymbol{\sigma}\in\{-1,1\}^n} e^{-\beta {\cal H}(\boldsymbol{\sigma};\textbf{h})}\,,\\
\widetilde f^{\rm pot}(m) &\equiv \frac{Jm^2}{2}- \frac1\beta \int dP_h(h)\ln\big(2\,{\rm cosh}\,[\beta\{Jm+h\}]\big) \,,
\end{align*}
where $\mathbb{E}$ is the expectation w.r.t. the random external fields $\textbf{h}$.
\begin{theorem}\label{thmCW_2}
The thermodynamic limit of the free energy for the Curie-Weiss model with random external fields verifies
\begin{align*}
\lim_{n\to\infty} f_n^{\rm rcw} =\inf_{m\in [-1, 1]}\widetilde f^{\rm pot}(m)\,.
\end{align*}	
\end{theorem}

Set $R_\epsilon(t)\equiv \epsilon+\int_0^tds\,m_\epsilon(s)$ with $m_\epsilon$ taking values in $[-1,1]$. The interpolating Hamiltonian and free energy are this time
\begin{align*}
{\cal H}_{t,\epsilon}(\boldsymbol{\sigma};\textbf{h}) &\equiv  (1-t) \Big(-\frac{J}{n} \sum_{i<j} \sigma_i\sigma_j\Big) - JR_\epsilon(t)\sum_{i=1}^n\sigma_i -\sum_{i=1}^nh_i\sigma_i\,,\\
f_{n,\epsilon}(t) &\equiv  - \frac1{\beta n}\mathbb{E}\ln {\cal Z}_{n,\epsilon}(t;\textbf{h}) = - \frac1{\beta n}\mathbb{E}\ln \sum_{\boldsymbol{\sigma}\in\{-1,1\}^n} e^{-\beta {\cal H}_{t,\epsilon}(\boldsymbol{\sigma};\textbf{h})}\,.
\end{align*}
By similar computations as before the sum rule becomes in this case (here the trial function $m_\epsilon$ is unconstrained as we did not use the concentration yet):
\begin{align}
f_n^{\rm rcw}
&=  - \frac1{\beta}\int dP_h(h)\ln\Big(2 \,{\rm cosh}\Big[\beta\Big\{h+J\int_0^1dt\, m_\epsilon(t)\Big\}\Big]\Big) \nonumber\\
&\qquad\quad- \int_0^1 dt\Big\{ \frac{J}{2}\mathbb{E}\langle M^2 \rangle_{t,\epsilon}-J\,\mathbb{E}\langle M \rangle_{t,\epsilon}m_\epsilon(t)\Big\} + O(s_n)\label{14}
\end{align}
or equivalently
\begin{align}
f_n^{\rm rcw}&=  \widetilde{f}^{\rm pot}\Big(\int_0^1 dt\,m_\epsilon(t)\Big) + O(s_n)\nonumber\\
&\qquad\quad+\frac J2\Big\{\int_0^1 dt\,m_\epsilon(t)^2 -\Big(\int_0^1 dt\,m_\epsilon(t) \Big)^2\Big\}- \frac J2\int_0^1dt\, \mathbb{E}\big\langle (M-m_\epsilon(t))^2\big\rangle_{t,\epsilon} \label{sr_cw_2}
\end{align}
where $O(s_n)$ is uniform in $t$, $\epsilon$, $R_{\epsilon}$.

From there a first bound is straightforward. Set $m_\epsilon(t)=m^*\equiv {\rm argmin}_{m\in[-1,1]} \widetilde f^{\rm pot}(m)$. Then the ``variance'' in curly brackets $\{\cdots\}$ in \eqref{sr_cw_2} cancels, while the ``remainder'' term $-\frac J2\int_0^1dt\,\mathbb{E}\langle (M-m^*)^2\rangle_{t,\epsilon} \le 0$, and thus $\limsup_{n\to\infty} f_n^{\rm rcw} \le \inf_{m \in[-1,1]} \widetilde{f}^{\rm pot}(m)$. Note that this first bound did not require the concentration of the magnetization, nor the use of the degree of freedom allowed by the possible time-dependence of the interpolation function $m_\epsilon(t)$. These two ingredients are used now for the converse bound. 

We now set $R_\epsilon$ to be the unique solution $R^*_{\epsilon,n}$ of the ODE $m_\epsilon(t)=F_n(t,R_{\epsilon}(t))=\mathbb{E}\langle M \rangle_{t,\epsilon}$ with initial condition $R_{\epsilon}(0)=\epsilon$; this solution exists for all $t\in [0,1]$ by the Cauchy-Lipschitz theorem. With this choice we check that the partial derivative $\partial_R F_n$ appearing in the exponential in the Liouville formula $\partial_\epsilon  R_{\epsilon,n}^*(t)=\exp\{\int_0^t ds\, \partial_{R_\epsilon} F_n(s,R_{\epsilon,n}^*(s))\}$ equals $\beta J n\mathbb{E}\langle (M - \langle M\rangle_{t,\epsilon})^2\rangle_{t,\epsilon}\ge 0$. Thus the flow $\epsilon \mapsto R^*_{\epsilon,n}(t)$ of the ODE is regular. We can then average the sum rule \eqref{sr_cw_2} over a small interval $\epsilon\in[s_n,2s_n]$ in order to use, thanks to the regularity of the flow, the following concentration for the magnetization (see the appendices for the proof): 
\begin{align}\label{magn_conc}
\frac{1}{s_n}\int_{s_n}^{2s_n}d\epsilon \int_0^1dt\,  \mathbb{E}\big\langle (M - \mathbb{E}\langle M\rangle_{t, \epsilon} )^2\big\rangle_{t,\epsilon} \le \frac{C(S,J)}{s_n^{2/3}n^{1/3}}
\end{align}
for a positive constant $C(S,J)$ depending only on the support of $P_h$ and the coupling strength $J$. This allows to simplify the sum rule by cancelling the remainder (up to a vanishing correction). Only the non-negative variance term survives which leads directly to (choosing a sequence $s_n$ going to $0_+$ at an appropriate rate) $f_n^{\rm rcw} \ge\widetilde{f}^{\rm pot}(\int_0^1dt\,m_{\epsilon,n}^*(t))+O(s_n)$. This finally yields $\liminf_{n\to\infty} f_n^{\rm rcw} \ge \inf_{m \in[-1,1]} \widetilde{f}^{\rm pot}(m)$, and thus proves Theorem \ref{thmCW_2}.
\section{Rank-one matrix estimation, or Wigner spike model}\label{sec:matrix}
Consider the following matrix estimation model, also called Wigner spike model in the statistics literature, which serves as a simple model of low-rank information extraction from a noisy data matrix such as in principal component analysis. One has access to the symmetric matrix of observations $\bW=(W_{ij})_{i,j=1}^n$ obtained as
\begin{align}\label{WSM}
W_{ij}= \frac{1}{\sqrt{n}}X_iX_j +  Z_{ij}\,, \qquad 1\le i\le j\le n\,,
\end{align}
where the signal-vector to infer has i.i.d. components $X_i\sim P_0$ with $\mathbb{E}[X_i^2] = \rho$, and the Gaussian noise is i.i.d. 
$Z_{ij}\sim{\cal N}(0,1)$ for $i\leq j$ and symmetric $Z_{ij}=Z_{ji}$. In order to ease the proof, we consider a prior supported on a bounded interval $[-S,S]$. Then, technical limiting arguments as found in \cite{2016arXiv161103888L,barbier2017phase} allow, if desired, to extend the final result to unbounded support as long as the first few moments of $P_0$ exist.

We consider the problem in the ``high-dimensional'' setting
where the total signal-to-noise ratio (SNR) per parameter: $\# \,{\rm observations} \cdot {\rm SNR_{obs} }/\# \,{\rm parameters\ to\ infer}$,
%
is an order one quantity, where $\rm SNR_{obs}$ denotes the SNR per observation. In the present case 
we have access to $n(n+1)/2$ independent observations and ${\rm SNR_{obs}}=\mathbb{E}[(X_1X_2)^2]/n = \rho^2/n$ for the $n(n-1)/2$ off-diagonal terms, ${\rm SNR_{obs}}=\mathbb{E}[X_1^4]/n$ for the diagonal ones.
%
%
Therefore we check $$\frac{(n(n-1)/2)\cdot(\rho^2/n)+n\cdot(\mathbb{E}[X_1^4]/n)}{n} = \frac{\rho^2}{2}+O(1/n)= O(1)\,.$$ This explains the presense of the scaling $1/\sqrt{n}$ in the observation model \eqref{WSM}. Note that any other scaling would make the estimation task either trivial if the total SNR per parameter tends to infinity, or impossible if it tends to zero.

We suppose that we are in a Bayesian optimal setting where the prior $P_0$ as well as the noise distribution are known. The posterior, or Gibbs distribution, is of the form $dP(\bx | \bW ) \propto \exp\{-\frac12\sum_{i\le j}(W_{ij}-x_ix_i/\sqrt{n})^2\}\prod_{i}dP_0(x_i)$. It is convenient to re-express it in terms 
of the independent variables $\bX$, $\bZ$ instead of $\bW$. Replacing $W_{ij}$ by its epression \eqref{WSM}, expanding the square, and then simplifying all the $\bx$-independent terms with the 
normalization, it becomes
\begin{align*}
dP(\bx | \bW(\bX,\bZ) ) = \frac{1}{{\cal Z}({\bf X}, {\bf Z})}\prod_{i=1}^n dP_0(x_i)\,e^{-{\cal H}({\bf x} ; {\bf X}, {\bf Z} )}\,,
\end{align*}
where the Hamiltonian and partition function are
\begin{align*}
{\cal H}({\bf x} ; {\bf X}, {\bf Z} )&\equiv\sum_{i\le j}^n \Big(\frac{x_i^2x_j^2}{2n}-\frac{x_ix_jX_iX_j}{n}-\frac{x_ix_jZ_{ij}}{\sqrt{n}}\Big)\,,\\
{\cal Z}_n(\bX, \bZ) &\equiv \int \prod_{i=1}^n dP_0(x_i) \,e^{-{\cal H}(\bx ; \bX,\textbf{Z})}\,.
\end{align*}
The free energy of this model is then defined as 
\begin{align*}
  f_n^{\rm ws} \equiv - \frac{1}{n}\mathbb{E}\ln {\cal Z}_n(\bX, \bZ) \,.
\end{align*}
Here $\mathbb{E}$ always denotes the expectation w.r.t. all (quenched) random variables in the ensuing expression (here $\bX$ and $\bZ$). This quantity is directly related to the mutual information between the observations and the input signal through the simple relation 
$$\frac{1}{n}I(\bX ; \bW) = f_n^{\rm ws} + \frac{\rho^2}{4}+O(1/n)$$ and variational expressions for this quantity are thus of fundamental interest. We will show that
such expressions can be rigorously determined using the adaptive interpolation method. 

The Hamiltonian ${\cal H}$ is nothing else than that of the so-called planted Sherrington-Kirkpatrick spin glass if 
one has a binary signal $X_i=\pm 1$ with Bernoulli $1/2$
prior; in this case the $\bx$-integral becomes a sum over spin configurations $\bx\in \{-1,1\}^n$. We shall see that, because 
this spin glass model stems from a Bayesian optimal setting, the {\it replica symmetric formula} for the free energy is exact. Let the replica symmetric potential be
\begin{align}
f^{\rm pot}(q,r)&\equiv \frac{q\,r}{2} - \frac{q^2}{4} - \mathbb{E} \ln \int d{P}_0(x) e^{-\big(r\,\frac{x^2}{2}-r\, x X -\sqrt{r}\,x Z\big)} \label{26}
\end{align}
with $X\sim P_0$, $Z\sim{\cal N}(0,1)$ and $(q,r)\in[0,\rho]^2$. This potential verifies at its stationary point(s)
\begin{align}\label{boundaries_WS}
\begin{cases}
\partial_{q} f^{\rm pot}(q,r) = 0 \quad \Leftrightarrow \quad r = q\,,\\
\partial_{r} f^{\rm pot}(q,r) = 0 \quad\Leftrightarrow\quad  q = \rho-{\rm mmse}(X|\sqrt{r}\, X + Z)\,,
\end{cases}
\end{align}
where, by definition, the minimum mean-square error (MMSE) function of a scalar r.v. $X\sim P_0$ observed through a Gaussian channel $Y= \sqrt{r}\, X + Z$, $Z\sim {\cal N}(0,1)$, with SNR $r$, is
\begin{align}
{\rm mmse}(X|\sqrt{r}\,X + Z) &\equiv \mathbb{E}_{X,Z}\big[(X-\mathbb{E}[X|\sqrt{r}\, X + Z])^2\big]\nonumber\\
 &=\mathbb{E}\Big[\Big(X-\frac{\int dP_0(x)\,x\, e^{-\big(r\,\frac{x^2}{2}-r\, x X -\sqrt{r}\,x Z\big)}}{\int dP_0(x) e^{-\big(r\,\frac{x^2}{2}-r\, x X -\sqrt{r}\,x Z\big)}}\Big)^2\Big]	\,.\label{scalarMMSE}
\end{align}
For completeness we prove by a direct computation the second identity of \eqref{boundaries_WS}, which also directly follows from the so-called ``I-MMSE'' theorem for the Gaussian channel \cite{GuoShamaiVerdu_IMMSE}, in the appendices.

We will prove the following theorem, already proved using the adaptive interpolation 
method in \cite{BarbierM17a} (formulated in a more technical discrete time setting)\footnote{The proof in reference \cite{BarbierM17a} is based on a ``discrete version'' of the adaptive interpolation method, in which the interpolation path is discretized and two different ``time parameters'' are used. The simpler (yet equivalent) continuous version presented in the present paper, as introduced in \cite{2017arXiv170910368B,barbier2017phase}, uses a single time parameter and is more straightforward for models defined on dense factor graphs (while the discrete version seems, at the moment, more adapted for sparse graphical models \cite{eric-censor-block}).}. 
Note that the theorem was also already proved in \cite{koradamacris} for a binary Bernoulli signal 
and also more recently using different (and more involved) techniques in \cite{XXT,2016arXiv161103888L,2018arXiv180101593E}. 

\begin{theorem}\label{thm:ws}
The thermodynamic limit of the free energy for the Wigner spike model verifies
\begin{align*}
\lim_{n\to\infty} f_n^{\rm ws} = {\adjustlimits \inf_{r\in [0,\rho]} \sup_{q\in [0,\rho]}} f^{\rm pot}(q,r)\,.
\end{align*}	
\end{theorem}

We note two remarks that are similar to those made for the Curie-Weiss model.
\begin{remark}\label{altxxx}
Here the maximum over $q$ is attained at $q=r$ and one finds
\begin{align*}
\lim_{n\to\infty} f_n^{\rm ws} 
= \inf_{r\in [0,\rho]}  f^{\rm pot}(r, r)
\end{align*}
as is usually found in the literature. This one-parameter variational expression can also be obtained directly by a simpler formulation of the adaptive interpolation discussed in section \ref{alternativexx}. For more complicated models, however, the simpler formulation is not powerful enough.
\end{remark}
\begin{remark}
The two-parameter potential \eqref{26} equals $qr/2$ minus the convex function $q^2/4$ minus another convex function of $r$. Using this structural property
it is possible to show (see the appendix D of \cite{barbier2017phase}) 
$$
{\adjustlimits \inf_{r\in [0,\rho]} \sup_{q\in [0,\rho]}} f^{\rm pot}(q,r)={\adjustlimits \inf_{q\in [0,\rho]} \sup_{r\in [0,\rho]}} f^{\rm pot}(q,r)\, .
$$
The r.h.s. can be optimized over $r$ by inverting the second equation in \eqref{boundaries_WS}. There is a unique solution $r(q)$ since the MMSE function is monotone decreasing in $r$, which yields
$$
\lim_{n\to\infty} f_n^{\rm ws} = \inf_{q\in [0,\rho]}  f^{\rm pot}(q,r(q))\, .
$$
In contrast to the Curie-Weiss model for general priors $P_0$ we do not have an explicit analytic expression for $r(q)$ and $f^{\rm pot}(q,r(q))$.
\end{remark}

Such replica formulas have also been proven for (non-symmetric) low-rank matrix and tensor estimation (or factorization) to varying degrees of generality \cite{koradamacris,2017arXiv170108010L,2017arXiv170200473M,2017arXiv170910368B,2018arXiv180101593E}, or in random linear \cite{barbier_ieee_replicaCS,barbier_allerton_RLE,private} and generalized estimation and learning \cite{barbier2017phase,NIPS2018_7584,NIPS2018_7453}. The proof reviewed here by the adaptive interpolation method is one of the simplest and most generic.
\subsection{Adaptive interpolation}\label{sec:adapInterp_XX}

Let $\epsilon \in[s_n, 2s_n]$, for some sequence $s_n\in(0,1/2]$ that tends to $0_+$ as $(1/2)n^{-\alpha}$ for $0<\alpha<1$ (as opposed to the Curie-Weiss model, here it is crucial that this perturbation $\epsilon >0$ as it plays the role of a SNR). Let $r_{\epsilon}: [0, 1] \mapsto [0,\rho]$ and set $R_\epsilon(t)\equiv \epsilon+\int_0^tds\,r_\epsilon(s)$.
Consider the following interpolating $t$-dependent estimation model, where $t\in[0,1]$, with accessible observations $(W_{ij}(t))_{i,j=1}^n$ and $(\widetilde W_i(t))_{i=1}^n$ obtained through
\begin{align*}
\begin{cases}
W_{ij}(t) =\sqrt{\frac{1-t}{n}}\, X_iX_j + Z_{ij}\,, \qquad &1\le i\le j\le n\,,\\
\widetilde{W}_i(t)  = \sqrt{R_{\epsilon}(t)}\,X_i + \widetilde Z_i\,, \qquad &1\le i\le n\,,
\end{cases} 
\end{align*}
with i.i.d. Gaussian noise $\widetilde Z_{i}\sim {\cal N}(0,1)$, and $Z_{ij}\sim {\cal N}(0,1)$ with $Z_{ij}=Z_{ji}$. The function $R_\epsilon$ thus plays the role of a SNR in a scalar (i.e., decoupled) denoising problem. The associated interpolating posterior written in Gibbs form is
\begin{align}\label{tpost}
dP_{t,\epsilon}(\bx|\bW(t;\bX, \bZ),\widetilde{\bW}(t;\bX, \widetilde\bZ))&=\frac{1}{\mathcal{Z}_{n,\epsilon}(t; \bX, \bZ, \bZt)}\prod_{i=1}^n dP_0(x_i)\, e^{-{\cal H}_{t,\epsilon}(\bx;\bX, \bZ, \bZt)} 
\end{align}
with normalization (i.e., partition function) $\mathcal{Z}_{n,\epsilon}(t; \bX, \bZ, \bZt)$ and interpolating Hamiltonian and free energy given by
\begin{align}
 {\cal H}_{t,\epsilon}(\bx;\bX, \bZ, \bZt)
&\equiv(1-t)\sum_{i\le j}^n \Big(\frac{x_i^2x_j^2}{2n}-\frac{x_ix_jX_iX_j}{n} -\frac{x_ix_jZ_{ij}}{\sqrt{n(1-t)}}\Big)\nonumber
\\ &\qquad\qquad+ R_\epsilon(t)\sum_{i=1}^n \Big(\frac{x_i^2}{2} - x_iX_i-\frac{x_i\widetilde{Z}_i}{\sqrt{R_\epsilon(t)}}\Big)\,,\label{Ht}\\
f_{n,\epsilon}(t)&\equiv-\frac{1}{n}\mathbb{E}\ln {\cal Z}_{n,\epsilon}(t; \bX, \bZ, \bZt)\, \label{fnt}.
\end{align}
The $t$-dependent Gibbs-bracket is defined as usual for functions $A(\bx)=A$
\begin{align*}
\langle A \rangle_{t,\epsilon} \equiv \int dP_{t,\epsilon}(\bx|\bW(t;\bX, \bZ),\widetilde{\bW}(t;\bX, \widetilde\bZ))\,A(\bx) \,.
\end{align*}
The interpolating free energy verifies the boundary conditions (here $X\sim P_0$, $Z\sim{\cal N}(0,1)$)
\begin{align}\label{bound2}
\begin{cases}
f_{n,\epsilon}(0)\hspace{-7pt} &= f_n^{\rm ws}+O(s_n)\,,\\
f_{n,\epsilon}(1) \hspace{-7pt}&=
-\mathbb{E}\ln\int dP_0(x)e^{-\big(R_{\epsilon}(t)\,\frac{x^2}{2}  -R_{\epsilon}(t)\, xX-\sqrt{R_{\epsilon}(t)}\,x Z\big)}\\
\hspace{-7pt}&=
-\mathbb{E}\ln\int dP_0(x)e^{-\big(\int_0^1 dt\,r_\epsilon(t)\,\frac{x^2}{2}  -\int_0^1 dt\,r_\epsilon(t)\, xX-\sqrt{\int_0^1dt\, r_\epsilon(t)}\,xZ\big)}+O(s_n)\,. 
\end{cases} 
\end{align}
To derive the first equality we use a computation that is 
almost identical to those found in the second appendix. Let us summarize it here. By the ``Nishimori'' identity\footnote{A direct consequence of the 
Bayes rule for conditional probabilities, see Lemma \ref{NishId} in the appendices.} and Gaussian integration by parts we get $|d f_{n, \epsilon}(0)/d\epsilon| = |\mathbb{E}\langle Q\rangle_{0,\epsilon}/2|\le \rho/2$, where $Q(\bx,\bX)=Q\equiv \frac{1}{n}\sum_{i=1}^nx_iX_i$ is the so-called overlap\footnote{The relevant order parameter of the problem.}. Again by the Nishimori identity $\mathbb{E}\langle Q\rangle_{t,\epsilon}\in[0,\rho]$. Therefore $\vert f_{n, \epsilon}(0) - f_{n,0}(0)\vert \le \rho s_n$ by the mean value theorem, and we have trivially that $f_{n,0}(0)=f_n^{\rm ws}$. The second equality of the second line is directly obtained by continuity.

The $t$-derivative of the free energy $f_{n,\epsilon}(t)$ can be computed and yields a formula analogous to 
\eqref{tder_ising}. The derivation here is a bit more involved. It requires again the use of Gaussian integration by parts w.r.t. $Z_{ij}$ and $\widetilde{Z}_i$ as well as the use of the Nishimori identity $\mathbb{E}[\langle x_i\rangle_{t,\epsilon} X_i] = 
\mathbb{E}[\langle x_i\rangle_{t,\epsilon}^2]$. Details for this computation are found in the appendices. One gets
\begin{align}
\frac{df_{n,\epsilon}(t)}{dt}= \frac{1}{4}\mathbb{E}\langle Q^2\rangle_{t,\epsilon} 
-  \frac{1}{2}\mathbb{E}\langle Q\rangle_{t,\epsilon} \,r_\epsilon(t) + O(1/n)\label{34}
\end{align}
where $O(1/n)$ depends only on $S$. Using the fundamental theorem of calculus and \eqref{bound2} we deduce the following fundamental sum rule:
\begin{align}\label{MF-sumrule}
f_n^{\rm ws} = -\mathbb{E}\ln\int dP_0(x)e^{-\big(\int_0^1 dt\,r_\epsilon(t)\,\frac{x^2}{2}-\int_0^1dt\, r_\epsilon(t)\, xX-\sqrt{\int_0^1dt\, r_\epsilon(t)}\,xZ\big)}\nonumber\\
\qquad\qquad\qquad   -  \frac{1}{4}\int_0^1dt \Big\{\mathbb{E}\langle Q^2\rangle_{t,\epsilon} -  2\,\mathbb{E}\langle Q\rangle_{t,\epsilon}  \,r_\epsilon(t)\Big\} + O(s_n)\,.
\end{align}
\subsection{Simplifying the sum rule: overlap concentration}
The ``perturbation'' $\epsilon$ forces the overlap to concentrate. Let us fix for concreteness $s_n=(1/2)n^{-1/16}$. We again say that the map $\epsilon\mapsto R_\epsilon(t)$ is regular if it is a ${\cal C}^1$ diffeomorphism whose Jacobian is greater or equal to one for all $t\in [0,1]$. Regularity implies
\begin{align}\label{over-concen}
\frac{1}{s_n}\int_{s_n}^{2s_n}d\epsilon \int_0^1dt\,  \mathbb{E}\big\langle (Q - \mathbb{E}\langle Q\rangle_{t, \epsilon} )^2\big\rangle_{t,\epsilon} \le \frac{C(S)}{n^{1/3}\,s_n^{4/3}} = \frac{C(S)2^{4/3}}{n^{1/4}}
\end{align}
for a positive constant $C(S)$ depending only on the support $S$ of the prior. As in the Curie-Weiss model, the perturbation cannot vanish too fast, which is enforced by the constraint $n^{1/3}\,s_n^{4/3}\to +\infty$ as $n\to +\infty$. 

Let us stress that the result as well as its proof are very generic and apply to essentially any Bayesian (optimal) inference problems perturbed by terms of the present type. We refer 
to the appendices for the proof and give here just a few comments. 
We must control two types of fluctuations: the thermal ones 
$\mathbb{E}\langle (Q - \langle Q\rangle_{t, \epsilon})^2\rangle_{t,\epsilon}$ and the quenched ones $\mathbb{E}[(\langle Q\rangle_{t, \epsilon} - \mathbb{E}\langle Q\rangle_{t, \epsilon})^2]$ (only thermal fluctuations are present in the pure Curie-Weiss model as there is no disorder/quenched variables).
The thermal fluctuations are again controlled by relating them to the second derivative of the free energy. 
We stress that this link is here non trivial and relies on Nishimori identities that are a direct consequence of the Bayesian 
optimal setting. It is precisely this link that guarantees the absence of replica symmetry breaking and the associated lack of self-averaging of the overlap \cite{MezardParisi87b}. 
On the other hand, the quenched fluctuations are small as a consequence of the concentration of the free energy, which itself is a very general fact. This aspect of the proof can be viewed as an adaptation of the Ghirlanda-Guerra identities of spin glasses \cite{ghirlanda1998general} to Bayesian inference problems (see \cite{KoradaMacris_CDMA}). 

Using this concentration result, under the regularity assumption for the map $\epsilon\mapsto R_\epsilon(t)$, the sum rule \eqref{MF-sumrule} simplifies to
\begin{align}\label{MF-sumrule-simple}
f_n^{\rm ws} = \frac1{s_n}\int_{s_n}^{2s_n}d\epsilon\Big[-\mathbb{E}\ln\int dP_0(x)e^{-\big(\int_0^1 dt\,r_\epsilon(t)\,\frac{x^2}{2}-\int_0^1 dt\,r_\epsilon(t)\, xX-\sqrt{\int_0^1dt \,r_\epsilon(t)}\,xZ\big)}\nonumber\\
\qquad\qquad\qquad   -  \frac{1}{4}\int_0^1 dt \Big\{\mathbb{E}[\langle Q\rangle_{t,\epsilon}]^2 -  2\,\mathbb{E}\langle Q\rangle_{t,\epsilon}  \,r_\epsilon(t)\Big\}\Big]+ O(s_n)
\end{align}
with a $O(s_n)$ depending only on $S$ (in particular it is uniform in $t$, $R_\epsilon$, $\epsilon$). The matching bounds are then obtained similarly to the Curie-Weiss model as explained below.
\subsection{Matching bounds}
\subsubsection{Upper bound}
Fix $r_{\epsilon}(t)= r\in[0,\rho]$ constant. Identifying the potential \eqref{26}, the sum rule \eqref{MF-sumrule-simple} then becomes
\begin{align*}
f_{n}^{\rm ws} &= \frac{1}{s_n}\int_{s_n}^{2s_n}d\epsilon\int_0^1dt\, f^{\rm pot}(\mathbb{E}\langle Q \rangle_{t,\epsilon},r)+O(s_n)\le \sup_{q\in[0,\rho]} f^{\rm pot}(q,r)+O(s_n)
\end{align*}
and optimizing over $r$ we obtain the desired bound $\limsup_{n\to\infty}f_{n}^{\rm ws}\le {\adjustlimits\inf_{r\in[0,\rho]}\sup_{q\in[0,\rho]}} f^{\rm pot}(q,r)$.
\subsubsection{Lower bound}
At this stage we choose $r_\epsilon(t)$ to be the solution of
\begin{align}
r_\epsilon(t)=\mathbb{E}\langle Q\rangle_{t,\epsilon} \,.\label{39}
\end{align}
As before, setting $F_n(t,R_{\epsilon}(t))=\mathbb{E}\langle Q\rangle_{t,\epsilon}$, we recognize 
a first order ODE
\begin{align}\label{odexx}
\frac{d R_\epsilon(t)}{dt} = F_n(t,R_{\epsilon}(t))\quad\text{with initial condition}\quad R_\epsilon(0)=\epsilon\,.
\end{align}
As $F_n(t,R_{\epsilon}(t))$ is ${\cal C}^1$ with bounded derivative w.r.t. its second argument the Cauchy-Lipschitz theorem implies that \eqref{odexx} admits a unique global solution $R_{\epsilon,n}^*(t)=\epsilon+\int_0^t ds\,r_{\epsilon,n}^*(s)$, where $r_{\epsilon,n}^*:[0,1]\mapsto [0,\rho]$. The flow $\epsilon\mapsto R_{\epsilon,n}^*(t)$ satisfies $\partial_\epsilon  R_{\epsilon,n}^*(t)=\exp\{\int_0^t ds\, \partial_{R_\epsilon} F_n(s,R_{\epsilon,n}^*(s))\}$ by Liouville's formula. Using repeatedly the Nishimori identity of Lemma \ref{NishId} one obtains
\begin{align*}
\frac{\partial F_n}{\partial R_\epsilon}(t,R_{\epsilon}(t))=\frac{1}{n}\sum_{i,j=1}^n\mathbb{E}\big[(\langle x_ix_j\rangle_{t,\epsilon}-\langle x_i\rangle_{t,\epsilon}\langle x_j\rangle_{t,\epsilon})^2\big]\ge 0
\end{align*}
so that the flow has a Jacobian $\ge 1$ and is a diffeomophism. Thus it is regular. This computation does not present any difficulty and can be found in section 6 of \cite{BarbierM17a}. Actually, even without doing this computation, one can directly assert that this derivative is non-negative using a simple information theoretic argument: the overlap, that quantifies the quality of the estimation, cannot decrease when the SNR (here $R_\epsilon$) increases, or equivalently the MMSE
\begin{align*}
{\rm MMSE}_{t,\epsilon} \equiv \min_{\hat \bx}\,\mathbb{E}\big[\|\bX - \hat \bx(\bW(t),\widetilde \bW(t))\|^2\big]	=\mathbb{E}\big[\|\bX - \langle \bx\rangle_{t,\epsilon}\|^2\big] = \rho - \mathbb{E}\langle Q\rangle_{t,\epsilon}
\end{align*}
cannot increase with the SNR. 

Explicit differentiation shows that $r\mapsto-\mathbb{E}\ln\int dP_0(x)\exp\{-(r\frac{x^2}{2}-rxX-\sqrt{r}\,xZ)\}$ is concave (see e.g., \cite{GuoShamaiVerdu}), so applying Jensen's inequality to the sum rule \eqref{MF-sumrule-simple} yields
\begin{align}\label{MF-sumrule_bound}
  f_{n}^{\rm ws} \ge \frac1{s_n}\int_{s_n}^{2s_n}d\epsilon \int_0^1 dt\,f^{\rm pot}(\mathbb{E}\langle Q\rangle_{t,\epsilon},r_{\epsilon,n}^*(t)) +O(s_n)\,.
\end{align}
Now note that
\begin{align}\label{remarkthat}
f^{\rm pot}(\mathbb{E}\langle Q\rangle_{t,\epsilon},r_{\epsilon,n}^*(t))=\sup_{q\in[0,\rho]}	f^{\rm pot}(q,r_{\epsilon,n}^*(t))\,.
\end{align}
This follows from the fact that $q\mapsto f^{\rm pot}(q,r)$ is concave (with second derivative equal to $-1/2$). Then because of \eqref{boundaries_WS} this map attains its maximum whenever $q=r$, which implies \eqref{remarkthat} because $r_{\epsilon,n}^*(t)$ verifies \eqref{39}. Thus \eqref{MF-sumrule_bound} gives
\begin{align*}
  f_{n}^{\rm ws} \ge \frac1{s_n}\int_{s_n}^{2s_n}d\epsilon \int_0^1dt  \sup_{q\in[0,\rho]}	f^{\rm pot}(q,r_{\epsilon,n}^*(t))+ O(s_n)\!\ge\! {\adjustlimits \inf_{r\in[0,\rho]} \sup_{q\in[0,\rho]}}	f^{\rm pot}(q,r) + O(s_n)
\end{align*}
and thus finally $\liminf_{n\to\infty}f_{n}^{\rm ws}\ge {\adjustlimits\inf_{r\in[0,\rho]}\sup_{q\in[0,\rho]}} f^{\rm pot}(q,r)$ which ends the proof of Theorem \ref{thm:ws}.
\section{Alternative formulation for matrix estimation}\label{alternativexx}
As for the Curie-Weiss model, we present for the sake of completeness a direct proof of the variational formula of remark \ref{altxxx} based on a single-parameter potential. Again, this route is simpler, but is less general and often not powerful enough. For example, in the non-symmetric version of the present problem, namely with observations of the type
\begin{align*}
W_{ij}= \frac{1}{\sqrt{n}}U_iV_j +  Z_{ij}\,, \qquad 1\le i\le j\le n\,,
\end{align*}
with $\textbf{U}$ and $\textbf{V}$ vectors being independently drawn from possibly different priors, the simpler path that we present now does not seem to generalize. See \cite{2017arXiv170910368B} for a treatment of this model (and its generalization to higher order tensors) by the adaptive interpolation method, or \cite{2017arXiv170200473M}.

We define this time the potential as \eqref{26} but where the stationary condition $r=q$ (recall \eqref{boundaries_WS}) is enforced:
\begin{align*}
\widetilde f^{\rm pot}(q)&\equiv f^{\rm pot}(q,q)= \frac{q^2}{4} - \mathbb{E} \ln \int d{P}_0(x) e^{-\big(q\frac{x^2}{2}-q x X -\sqrt{q}\,x Z\big)}\,.  
\end{align*}
\begin{theorem}\label{thm:ws_2}
The thermodynamic limit of the free energy for the Wigner spike model verifies
\begin{align*}
\lim_{n\to\infty} f_n^{\rm ws} =\inf_{q\in [0,\rho]}  \widetilde f^{\rm pot}(q)\,.
\end{align*}	
\end{theorem}

Set $R_\epsilon(t)\equiv \epsilon+\int_0^tds\, q_\epsilon(s)$ with $q_\epsilon$ taking values in $[0,\rho]$ and the Hamiltonian ${\cal H}_{t,\epsilon}$ given by \eqref{Ht} remains unchanged. Then, obviously, the sum rule \eqref{MF-sumrule} stays the same too (simply renaming $r_\epsilon$ by $q_\epsilon$). It can be equivalently re-expressed as
\begin{align*}
f_n^{\rm ws} = &\ \widetilde f^{\rm pot}\Big(\int_0^1 dt\,q_{\epsilon}(t)\Big)+\frac14\Big\{\int_0^1 dt\,q_\epsilon(t)^2 -\Big(\int_0^1dt\, q_\epsilon(t) \Big)^2\Big\}\nonumber\\
&\qquad  -  \frac{1}{4}\int_0^1dt\, \mathbb{E}\big\langle (Q-  q_\epsilon(t))^2\big\rangle_{t,\epsilon} + O(s_n)\,.
\end{align*}
From there the steps follow exactly the ones presented in section \ref{anotherpath_cw} for the Curie-Weiss model. The upper bound is obtained by setting  $q_\epsilon(t)=q^*\equiv {\rm argmin}_{q\in[0,\rho]} \widetilde f^{\rm pot}(q)$ which cancels the variance term in curly brackets $\{\cdots\}$, and then using the non-positivity of the remainder. The lower bound is obtained from the choice of the interpolation function as solution of the ODE $q_\epsilon(t)=\mathbb{E}\langle Q\rangle_{t,\epsilon}$ (checking carefully that its flow is regular). Then one has to average the sum rule over $\epsilon\in[s_n,2s_n]$ in order to use the overlap concentration. Finally, using the non-negativity of the variance term ends the argument. This ends the proof of Theorem \ref{thm:ws_2}.
\section{Conclusion and perspectives}
We have presented a proof technique allowing to rigorously derive replica symmetric formulas for the free energy of statistical mechanics and Bayesian inference models. We focused on the simplest possible ones for pedagogic purpose. 
Our approach appears to be much more compact and straightforward than other existing techniques
when looking at the variety of problems successfully treated with it. In addition it only requires, in order to be applicable, what is believed to be from a physical point a view the minimal property for replica {\it symmetric} formulas to be valid: concentration of the order parameter of the problem (the magnetization in Curie-Weiss, or more generally the Edwards-Anderson overlap). Nevertheless the adaptive interpolation method is, at the moment, restricted to simple physics models, or Bayesian {\it optimal} inference and learning problems on {\it dense} graphs, for which such concentration can be proven in the whole phase diagram. Extending this technique to more complicated models away from the Nishimori line (i.e., in the non-optimal Bayesian setting of inference problems when the posterior is not exactly known), to problems where the self-averaging of the overlap does not occur and replica symmetry breaking takes place, and also to sparse graphical models, are exciting and challenging research directions.

Let us mention another interesting line of work initiated by Guerra \cite{Guerra_sumRUles} exploiting interpolation methods based on Hamilton-Jacobi partial differential equations satisfied by the free energy \cite{barra2013mean,doi:10.1063/1.3131687,1742-5468-2010-09-P09006,Barra2008} (and more recently \cite{mourrat2018hamilton} which treats the Wigner spike model). We believe that there exist profound links between this approach and ours, although their practical implementation is quite different. Elucidating these connections seems an interesting research direction.
\begin{acknowledgements}
This work was developed while Jean Barbier was affiliated with EPFL. Support from Swiss National Foundation grant no 200021-156672 is acknowledged.
It is a pleasure to thank our collaborators with whom we further developed the method for various applications in other problems, in particular Chun-Lam Chan, Florent Krzakala,  Cl\'ement Luneau, Antoine Maillard, L\'eo Miolane and Lenka Zdeborov\'a.
\end{acknowledgements}
%
\section*{Appendices} 

\subsection*{(A1) The Nishimori identity}\label{app:nishimori}
\begin{lemma}\label{NishId}
Let $(\bX,\bY)$ be a couple of random variables with joint distribution $P(\bX, \bY)$ and conditional distribution 
$P(\bX | \bY)$. Let $k \geq 1$ and let $\bx^{(1)}, \dots, \bx^{(k)}$ be i.i.d.\ samples from the conditional distribution. Let us denote $\langle - \rangle$ the expectation operator w.r.t. the conditional distribution and $\mathbb{E}$ the expectation w.r.t. the joint distribution. Then, for all continuous bounded function $g$ we have\footnote{This identity has been abusively called ``Nishimori identity'' in the statistical physics literature despite that
it is a simple consequence of Bayes formula. The 
``true'' Nishimori identity concerns models with one extra feature, namely a 
gauge symmetry which allows to eliminate the input signal, and the expectation over the signal $\bX$ in expressions of the form $\mathbb{E}\langle -\rangle$
can therefore be dropped (see e.g., \cite{koradamacris}).} 
\begin{align*}
\mathbb{E} \big\langle g(\bY,\bx^{(1)}, \dots, \bx^{(k)}) \big\rangle
=
\mathbb{E} \big\langle g(\bY, \bX, \bx^{(2)}, \dots, \bx^{(k)}) \big\rangle\,. 
\end{align*}	
\end{lemma}
\begin{proof}
This is a simple consequence of Bayes formula.
It is equivalent to sample the couple $(\bX,\bY)$ according to its joint distribution or to sample first $\bY$ according to its marginal distribution and then to sample $\bx$ conditionally on $\bY$ from the conditional distribution. Thus the two $(k+1)$-tuples $(\bY,\bx^{(1)}, \dots,\bx^{(k)})$ and  $(\bY, \bX, \bx^{(2)},\dots,\bx^{(k)})$ have the same law.	
\end{proof}
\subsection*{(A2) Concentration of the magnetization: proof of inequality \eqref{magn_conc}}\label{appendix-magn}
We prove the concentration property \eqref{magn_conc} under the assumption that the map $\epsilon\in [s_n, 2s_n] \mapsto R_\epsilon(t)\in [R_{s_n}(t), R_{2s_n}(t)]$ is regular. Recall regularity here means that the map is a $\mathcal{C}^1$ diffeomorphism 
 with Jacobian $\partial_\epsilon R_\epsilon(t) \geq 1$. We do not repeat this assumption in the statements below.

To control the total fluctuations of the magnetization we use the decomposition
\begin{align*}
\mathbb{E}\big\langle (M - \mathbb{E}\langle {M}\rangle_{t,\epsilon})^2\big\rangle_{t,\epsilon}
& = 
\mathbb{E}\big\langle ({M} - \langle {M}\rangle_{t,\epsilon})^2\big\rangle_{t,\epsilon}
+ 
\mathbb{E}\big[(\langle {M}\rangle_{t,\epsilon} - \mathbb{E}\langle {M}\rangle_{t,\epsilon})^2\big]
\end{align*} 
where the first fluctuations are thermal while the second ones are fluctuations due to the quenched disorder. In this appendix the $t$-dependence of the free energy and of $R_{\epsilon}(t)$ does not play any role, so we drop it and simply denote $R_{\epsilon}\equiv R_{\epsilon}(t)$; the proof is valid at any $t\in[0,1]$. We emphasize the $R_\epsilon$-dependence of the interpolating free energy and shall denote
$f_{n,\epsilon}(t, R_\epsilon(t)) =f_{n,\epsilon}(R_\epsilon) \equiv -\frac{1}{\beta n}\mathbb{E}\ln {\cal Z}_{n,\epsilon}(t;\textbf{h})$ and $F_{n,\epsilon}(R_\epsilon) \equiv -\frac{1}{\beta n}\ln {\cal Z}_{n,\epsilon}(t;\textbf{h})$ for the averaged and  
non-averaged free energies. The derivatives of these free energies satisfy
\begin{align}
\frac{d f_{n,\epsilon}}{d R_\epsilon}=-J\,\mathbb{E}\langle M\rangle_{t,\epsilon}\,, 
\qquad \frac{d^2 f_{n,\epsilon}}{d R_\epsilon^2}  = - n\beta J^2 \,\mathbb{E}\big\langle (M - \langle M\rangle_{t, \epsilon})^2\big\rangle_{t,\epsilon}\label{der_averaged}
\end{align}
and the same relations, but {\it without} the expectation $\mathbb{E}$ over $\bf h$, hold for the derivatives of $F_{n,\epsilon}$ (e.g., the first relation becomes $d F_{n,\epsilon}/d R_\epsilon=-J\,\langle M\rangle_{t,\epsilon}$ and similarly for the second one).

The thermal fluctuations are controlled using the 
same integration arguments as in section \ref{sec:2.2} for the Curie-Weiss model with constant external field, based on the relations \eqref{der_averaged}.
This yields the random field counterpart of \eqref{concen-cw}:
\begin{lemma}\label{thermal-fluctuations-rcw}
	We have $\int_{s_n}^{2s_n}d\epsilon \,  \mathbb{E}\big\langle (M - \langle M\rangle_{t, \epsilon} )^2\big\rangle_{t,\epsilon} \le \frac{2}{\beta J n}$.
\end{lemma}

The next lemma takes care of the fluctuations due to the random external field:
\begin{lemma}\label{diseorder-fluctuations-rcw}
	There exists $C(S,J)>0$ s.t. $\int_{s_n}^{2s_n}d\epsilon \,  \mathbb{E}\big[(\langle M\rangle_{t,\epsilon} - \mathbb{E}\langle M\rangle_{t, \epsilon} )^2\big] \le C(S,J)(s_n/n)^{1/3}$.
\end{lemma}
\begin{proof}
The proof is based on relating the derivatives of the free energy to the magnetization, and then using that the free energy concentrates which itself is a very generic result. Relations \eqref{der_averaged} show that the averaged and non-averaged free energies are concave in $R_\epsilon$. This
concavity allows to use the following lemma (see e.g., \cite{BarbierM17a,barbier2017phase} for a proof):
\begin{lemma}\label{lemmaConvexity}
Let $G(x)$ and $g(x)$ be concave functions. Let $\delta>0$ and define $C^{-}_\delta(x) \equiv g'(x-\delta) - g'(x) \geq 0$ and $C^{+}_\delta(x) \equiv g'(x) - g'(x+\delta) \geq 0$. Then
\begin{align*}
|G'(x) - g'(x)| \leq \delta^{-1} \sum_{u \in \{x-\delta,\, x,\, x+\delta\}} |G(u)-g(u)| + C^{+}_\delta(x) + C^{-}_\delta(x)\,.
\end{align*}
\end{lemma}
Applied to $G(x) \to F_{n,\epsilon}(R_\epsilon)$ and $g(x)\to f_{n,\epsilon}(R_\epsilon)$ with $x\to R_\epsilon$ 
this gives, using \eqref{der_averaged},
\begin{align*}
J\vert \langle {M}\rangle_{t,\epsilon} - \mathbb{E}\langle {M}\rangle_{t,\epsilon}\vert & \leq 
\delta^{-1} \sum_{u\in \{R_\epsilon -\delta,\, R_\epsilon,\, R_\epsilon+\delta\}}
 \vert F_{n, \epsilon}(u) - f_{n, \epsilon}(u) \vert + C_\delta^+(R_\epsilon) 
  + C_\delta^-(R_\epsilon) 
\end{align*}
where $C_\delta^-(R_\epsilon)\equiv f_{n,\epsilon}'(R_\epsilon-\delta)- f_{n,\epsilon}'(R_\epsilon)\ge 0$ 
and $C_\delta^+(R_\epsilon)\equiv f_{n,\epsilon}'(R_\epsilon)-f_{n,\epsilon}'(R_\epsilon+\delta)\ge 0$ 
where the derivative $f_{n,\epsilon}'$ is a $R_\epsilon$-derivative. We now square this identity and take its expectation. 
Then using $(\sum_{i=1}^pv_i)^2 \le p\sum_{i=1}^pv_i^2$ by convexity it yields
\begin{align}
\frac{J^2}{5}\mathbb{E}\big[(\langle {M}\rangle_{t,\epsilon} - \mathbb{E}\langle {M}\rangle_{t,\epsilon})^2\big]&\leq 
\delta^{-2} \sum_{u\in \{R_\epsilon -\delta,\, R_\epsilon,\, R_\epsilon+\delta\}} \mathbb{E}\big[(F_{n, \epsilon}(u) - f_{n, \epsilon}(u))^2\big] 
\nonumber\\&
\qquad\qquad\qquad\qquad\qquad\qquad+ C_\delta^+(R_\epsilon)^2 + C_\delta^-(R_\epsilon)^2\,. \label{stuff}
\end{align}
At this stage we need a crucial result on the concentration of the free energy (recall $|h_i|\le S$):
\begin{proposition}\label{concentrationtheorem}
There exists $C(S)>0$ s.t. $\mathbb{E}\big[\{ F_{n, \epsilon}(t,R_\epsilon(t)) - f_{n,\epsilon}(t,R_\epsilon(t))\}^2\big] 
\leq 
\frac{C(S)}{n}$.
\end{proposition}
This result is very generic and is easily proven using standard methods (for example an application of the Efron-Stein 
inequality will give the result directly; this is where it is convenient to assume bounded fields $|h_i|<S$). 
Therefore the free energy differences in the sum in \eqref{stuff} are small. It remains to control the two other terms. 
We have the crude bound $0\leq C^{\pm}_\delta\le2\sup|f_{n,\epsilon}^\prime| = 2\,J$ from \eqref{der_averaged}. Therefore 
\begin{align*}
 \int_{s_n}^{2s_n} d\epsilon\, \big\{C_\delta^+(R_{\epsilon})^2 + C_\delta^-(R_{\epsilon})^2\big\} 
 &\leq
 2\,J
 \int_{s_n}^{2s_n} d\epsilon\, \big\{C_\delta^+(R_{\epsilon}) + C_\delta^-(R_{\epsilon})\big\}
 \nonumber \\ 
  &\leq 
  2\,J
 \int_{R_{s_n}}^{R_{2s_n}} dR_\epsilon\, \big\{C_\delta^+(R_\epsilon) + C_\delta^-(R_\epsilon)\big\}
 \nonumber \\ 
&\hspace{-3.5cm}= 
2\,J\Big[\Big( f_{n,\epsilon}(R_{s_n}+\delta) -  f_{n,\epsilon}(R_{s_n}-\delta)\Big) 
+ \Big(f_{n,\epsilon}(R_{2s_n}-\delta) - f_{n,\epsilon}(R_{2s_n}+\delta)\Big)\Big]
\end{align*}
where we used that the Jacobian of the ${\cal C}^1$ diffeomorphism $\epsilon\mapsto R_\epsilon(t)$ is $\ge 1$ (by regularity) for the second inequality. 
The mean value theorem and $|f_{n,\epsilon}^\prime| \le J$ imply $|f_{n, \epsilon}(R_\epsilon-\delta) - f_{n,\epsilon}(R_\epsilon + \delta)|\le 2\delta J$. Thus we obtain
\begin{align*}
 \int_{s_n}^{2s_n} d\epsilon\, \big\{C_\delta^+(R_{\epsilon})^2 + C_\delta^-(R_{\epsilon})^2\big\}\leq 
 8\delta J \,.
\end{align*}
Finally, integrating \eqref{stuff} over $\epsilon\in [s_n, 2s_n]$ yields
\begin{align*}
 &\int_{s_n}^{2s_n} d\epsilon\, 
 \mathbb{E}\big[(\langle {M}\rangle_{t,\epsilon} - \mathbb{E}\langle {M}\rangle_{t,\epsilon})^2\big]\leq \frac{15C(S)}{J^2}\frac{s_n}{n\delta^2}  + \frac{40}{J}\delta   \,.
\end{align*}
The bound is optimized choosing $\delta  =  (s_n/n)^{1/3}$. This ends the proof of Lemma \ref{diseorder-fluctuations-rcw}.	
\end{proof}

Combining Lemmas \ref{thermal-fluctuations-rcw} and
\ref{diseorder-fluctuations-rcw} together with Fubini's theorem yields \eqref{magn_conc}.
\subsection*{(A3) Proof of \eqref{boundaries_WS}}
Let us show how to obtain the second stationary condition in \eqref{boundaries_WS}. It follows from the I-MMSE theorem \cite{GuoShamaiVerdu_IMMSE}, but we prove it for completeness. Starting from $\partial_{r} f^{\rm pot}(q,r) = 0$ gives
\begin{align*}
q 	= 2\frac{d}{dr} \mathbb{E} \ln \int d{P}_0(x) e^{-\big(r\,\frac{x^2}{2}-r\, x X -\sqrt{r}\,x Z\big)}= \mathbb{E}\Big\langle -x^2+2xX+\frac{xZ}{\sqrt{r}}\Big\rangle_r
\end{align*}
where $\langle - \rangle_r$ is the expectation operator w.r.t. the probability density over $x$ that is proportional to $P_0(x) \exp\{-(r\,x^2/2-r\, x X -\sqrt{r}\,x Z)\}$. Now we integrate by part the Gaussian variable using the elementary formula $\mathbb{E}[Zg(Z)]=\mathbb{E}\,g'(Z)$ for $Z\sim {\cal N}(0,1)$ and $g$ a bounded function. This gives $\mathbb{E}[Z\langle x\rangle_r]/\sqrt{r}=\mathbb{E}[\langle x^2\rangle_r-\langle x\rangle_r^2]$. Therefore $$\mathbb{E}\Big\langle -x^2+2xX+\frac{xZ}{\sqrt{r}}\Big\rangle_r = 2\mathbb{E}[X\langle x\rangle_r]-\mathbb{E}[\langle x\rangle_r^2]=\mathbb{E}[X\langle x\rangle_r]$$ where we used the Nishimori identity Lemma~\ref{NishId}, which implies $\mathbb{E}[X\langle x\rangle_r]=\mathbb{E}[\langle x\rangle_r^2]$, for the last equality. Developing ${\rm mmse}(X|\sqrt{r}\,X + Z)$ given by \eqref{scalarMMSE} we get
\begin{align*}
{\rm mmse}(X|\sqrt{r}\,X + Z)=\mathbb{E}[X^2]-2\mathbb{E}[ X\langle x \rangle_r]+\	\mathbb{E}[\langle x \rangle_r^2] = \rho-\mathbb{E}[ X\langle x \rangle_r]
\end{align*}
recalling $\mathbb{E}[X^2]=\rho$ as $X\sim P_0$, and using again the same Nishimori identity. Combining everything finishes the proof.

\subsection*{(A4) Free energy $t$-derivative: proof of identity \eqref{34}}\label{compDerivative}
Let us compute $\frac{df_{n,\epsilon}(t)}{dt}=\frac{1}{n}\mathbb{E}\big\langle\frac{d {\cal H}_{t,\epsilon}}{dt}\big\rangle_{t,\epsilon}$. Starting from \eqref{Ht}, \eqref{fnt} one obtains 
\begin{align*}
&\frac{df_{n,\epsilon}}{dt}\!=\!\frac{1}{n}\mathbb{E}\Big\langle r_\epsilon(t)\sum_{i=1}^n\!\Big(\frac{x_i^2}{2} - x_iX_i-\frac{x_i\widetilde Z_i}{2\sqrt{R_\epsilon(t)}} \Big) - \sum_{i\le j}^n\!\Big(\frac{x_i^2x_j^2}{2n} - \frac{x_ix_jX_iX_j}{n}-\frac{x_ix_j Z_{ij}}{2\sqrt{n(1-t)}} \Big) \Big\rangle_{t,\epsilon}.
\end{align*}
Now we integrate by part the Gaussian noise using, as in the previous appendix, $\mathbb{E}[Zg(Z)]=\mathbb{E}\,g'(Z)$ for $Z\sim {\cal N}(0,1)$. This leads to
\begin{align*}
\frac{df_{n,\epsilon}}{dt}= \frac{1}{n}\mathbb{E}\Big\langle r_\epsilon(t)\sum_{i=1}^n\Big(\frac{x_ix_i'}{2}- x_iX_i\Big)- \sum_{i\le j}^n\Big(\frac{x_ix_j x_i'x_j'}{2n}- \frac{x_ix_jX_iX_j}{n} \Big) \Big\rangle_{t,\epsilon}
\end{align*}
where $\bx$, $\bx'$ are two i.i.d. replicas drawn according to the posterior \eqref{tpost} (and thus share the same quenched variables). Slightly abusing notation, we continue to use the notation $\langle A(\bx,\bx')\rangle_{t,\epsilon}$ for the expectation w.r.t. the product measure $dP_{t,\epsilon}(\bx|\bW,\widetilde \bW)dP_{t,\epsilon}(\bx'|\bW,\widetilde \bW)$ of the i.i.d. replicas. An application of the Nishimori identities $\mathbb{E}[\langle x_i\rangle_{t,\epsilon} X_i] = \mathbb{E}\langle x_i x_i^\prime\rangle_{t,\epsilon}$ and $\mathbb{E}[\langle x_ix_j\rangle_{t,\epsilon} X_iX_j] = \mathbb{E}\langle x_ix_j x_i^\prime x_j^\prime\rangle_{t,\epsilon}$ then yields
\begin{align}
\frac{df_{n,\epsilon}}{dt}&= \frac{1}{2}\mathbb{E}\Big\langle \frac{1}{n^2} \sum_{i\le j}^nx_ix_jX_iX_j -\frac{r_\epsilon(t)}{n}\sum_{i=1}^n x_iX_i  \Big\rangle_{t,\epsilon} \nonumber\\
&= \frac{1}{2}\mathbb{E}\Big\langle  \frac{1}{2n^2} \sum_{i,j=1}^n x_ix_jX_iX_j +\frac{1}{2n^2} \sum_{i=1}^n x_i^2X_i^2 -\frac{r_\epsilon(t)}{n}\sum_{i=1}^n x_iX_i \Big\rangle_{t,\epsilon}\,. \label{last}
\end{align}
The Cauchy-Schwarz inequality combined with the Nishimori identity imply that, as long as $P_0$ has bounded fourth moment, $\mathbb{E}\langle \frac{1}{n^{2}}\sum_{i=1}^n x_i^2X_i^2\rangle_{t,\epsilon} =O(1/n)$ . Indeed, by Cauchy-Schwarz
\begin{align*}
 \mathbb{E}\Big\langle \frac1n \sum_{i=1}^n x_i^2X_i^2\Big\rangle_{t,\epsilon} \leq 
 \Big\{\mathbb{E}\Big\langle \frac1n\sum_{i=1}^n x_i^4\Big\rangle_{t,\epsilon}\,\mathbb{E}\Big[\frac1n\sum_{i=1}^n X_i^4\Big]\Big\}^{1/2}= \mathbb{E}_{P_0}[X_1^4]
\end{align*}
where we used the Nishimori identity $\mathbb{E}\langle x_i^4\rangle_{t,\epsilon} = \mathbb{E}[ X_i^4]$ and the independence of the signal entries for getting the last equality. Expressing the two other terms in \eqref{last} with the help of the overlap $Q=\frac1n\sum_{i=1}^n x_iX_i$ gives \eqref{34}.
\subsection*{(A5) Concentration of the overlap: proof of inequality \eqref{over-concen}}\label{appendix-overlap}
As before, here the $t$-dependence will not play a role and the whole argument applies to any $t\in[0,1]$. We then simply 
denote $R_{\epsilon}\equiv R_{\epsilon}(t)$ and $f_{n,\epsilon}(t, R_\epsilon(t)) =f_{n,\epsilon}(R_\epsilon)$. 
The concentration property \eqref{over-concen} is again proven under the assumption 
that the map $\epsilon\in [s_n, 2s_n] \mapsto R_\epsilon(t)\in [R_{s_n}(t), R_{2s_n}(t)]$ is regular, and we 
do not repeat this assumption in the statements below. The proof is quite generic and the general ideas apply also to other problems. 

Let 
\begin{align}
\mathcal{L}(\bx,\bX,\bZ) =\mathcal{L} \equiv \frac{1}{n}\sum_{i=1}^n\Big(\frac{x_i^2}{2} - x_i X_i - \frac{x_i \widetilde Z_i}{2\sqrt{R_\epsilon(t)}} \Big)\,.\label{def_L}
\end{align}
Up to the prefactor $1/n$ this quantity is the $R_\epsilon$-derivative of ${\cal H}_{t,\epsilon}$ given by \eqref{Ht}. The overlap fluctuations are upper bounded by those of $\mathcal{L}$, which are easier to control, as
%
\begin{align}
\mathbb{E}\big\langle (Q - \mathbb{E}\langle Q \rangle_{t,\epsilon})^2\big\rangle_{t,\epsilon} \le 4\,\mathbb{E}\big\langle (\mathcal{L} - \mathbb{E}\langle \mathcal{L}\rangle_{t,\epsilon})^2\big\rangle_{t,\epsilon}\,.\label{remarkable}
\end{align}
A detailed derivation can be found in the last appendix and involves only elementary algebra using the Nishimori identity
and integrations by parts w.r.t.\ the Gaussian noise $\widetilde Z_i$. The concentration inequality \eqref{over-concen} is then a direct consequence of the following result (combined with Fubini's theorem):
\begin{proposition} There exists $C(S)>0$ s.t. $\int_{s_n}^{2s_n} d\epsilon\, 
\mathbb{E}\big\langle (\mathcal{L} - \mathbb{E}\langle \mathcal{L}\rangle_{t,\epsilon})^2\big\rangle_{t,\epsilon} \le C(S)(s_nn)^{-1/3}$.\label{L-concentration}
\end{proposition}

The proof of this proposition is broken in two parts, using again the decomposition
\begin{align*}
\mathbb{E}\big\langle (\mathcal{L} - \mathbb{E}\langle \mathcal{L}\rangle_{t,\epsilon})^2\big\rangle_{t,\epsilon}
& = 
\mathbb{E}\big\langle (\mathcal{L} - \langle \mathcal{L}\rangle_{t,\epsilon})^2\big\rangle_{t,\epsilon}
+ 
\mathbb{E}\big[(\langle \mathcal{L}\rangle_{t,\epsilon} - \mathbb{E}\langle \mathcal{L}\rangle_{t,\epsilon})^2\big]\,.
\end{align*}
Thus it suffices to prove the two following lemmas. The first lemma expresses concentration w.r.t.\ the posterior distribution (or ``thermal fluctuations'') and is a direct consequence of concavity properties of the average free energy and the Nishimori identity.
\begin{lemma}\label{thermal-fluctuations}
	We have $\int_{s_n}^{2s_n} d\epsilon\, 
  \mathbb{E} \big\langle (\mathcal{L} - \langle \mathcal{L}\rangle_{t,\epsilon})^2 \big\rangle_{t,\epsilon}  \le \frac{\rho}{n}(1+\frac{\ln2}{4})$.
\end{lemma}
\begin{proof}

We emphasize again that the interpolating free energy \eqref{fnt} is here viewed as a function of $R_\epsilon$. In the argument that follows we consider derivatives of this function w.r.t. $R_\epsilon$.
By direct computation
\begin{align}
\mathbb{E}\big\langle (\mathcal{L} - \langle \mathcal{L} \rangle_{t,\epsilon})^2\big\rangle_{t,\epsilon}
& = 
-\frac{1}{n}\frac{d^2f_{n,\epsilon}}{dR_{\epsilon}^2}
+\frac{1}{4n^2R_{\epsilon}} \sum_{i=1}^n \mathbb{E}\big[\langle x_i^2\rangle_{t,\epsilon} - \langle x_i\rangle_{t,\epsilon}^2\big] 
\nonumber \\ &
\leq 
-\frac{1}{n}\frac{d^2f_{n,\epsilon}}{dR_{\epsilon}^2} +\frac{\rho}{4n\epsilon} \,,
\label{directcomputation}
\end{align}
where we used $R_{\epsilon}\geq \epsilon$ and $\mathbb{E}\langle x_i^2\rangle_{t,\epsilon} = \mathbb{E}[X_i^2]=\rho$ (an application of the Nishimori identity). 
We integrate this inequality over $\epsilon\in [s_n, 2s_n]$. Recall the map 
$\epsilon\in [s_n, 2s_n]\mapsto R_{\epsilon}(t)\in [R_{s_n}(t), R_{2s_n}(t)]$ has a Jacobian $\ge 1$, is ${\cal C}^1$ 
and has a well defined ${\cal C}^1$ inverse since we have assumed that it is regular. 
Thus integrating \eqref{directcomputation} and performing a change of variable (to get the second inequality) we obtain
\begin{align*}
\int_{s_n}^{2s_n} d\epsilon\, \mathbb{E}\big\langle (\mathcal{L} - \langle \mathcal{L} \rangle_{t,\epsilon})^2\big\rangle_{t,\epsilon}
& \leq 
- \frac{1}{n}\int_{s_n}^{2s_n} d\epsilon \,\frac{d^2f_{n,\epsilon}}{dR_{\epsilon}^2} + \frac{\rho}{4n}\int_{s_n}^{2s_n} \,\frac{d\epsilon}{\epsilon} 
\nonumber \\ &
\leq 
- \frac{1}{n}\int_{R_{s_n}}^{R_{2s_n}} dR_{\epsilon} \,\frac{d^2f_{n,\epsilon}}{dR_{\epsilon}^2}
+ \frac{\rho}{4n}\int_{s_n}^{2s_n} \,\frac{d\epsilon}{\epsilon} 
\nonumber \\ &
=
\frac{1}{n}\Big(\frac{df_{n,\epsilon}}{dR_{\epsilon}}(R_{s_n}) 
- \frac{df_{n,\epsilon}}{dR_{\epsilon}}(R_{2s_n})\Big)+ \frac{\rho}{4n}\ln 2\,.
\end{align*}
We have $|d f_{n, \epsilon}/dR_\epsilon| = |\mathbb{E}\langle Q\rangle_{t,\epsilon}/2|\le \rho/2$ so the first term is certainly smaller in absolute value than $\rho/n$. This concludes the proof of Lemma \ref{thermal-fluctuations}.
\end{proof}

The second lemma expresses the concentration w.r.t.\ the quenched disorder variables
and is a consequence of the concentration of the free energy onto its average (w.r.t. the quenched variables).
\begin{lemma}\label{disorder-fluctuations}
	There exists $C(S)>0$ s.t. $\int_{s_n}^{2s_n} d\epsilon\, 
  \mathbb{E}\big[ (\langle \mathcal{L}\rangle_{t,\epsilon} - \mathbb{E}\langle \mathcal{L}\rangle_{t,\epsilon})^2\big] \le C(S)(s_n\,n)^{-1/3}$.
\end{lemma}
\begin{proof}
The proof is similar than the one of Lemma \ref{diseorder-fluctuations-rcw} for the Curie-Weiss model in a random field, but 
with some more subtleties. Recall the free energies, seen as functions of $R_\epsilon\equiv R_\epsilon(t)$ are
$f_{n, \epsilon}(t,R_\epsilon(t))=f_{n, \epsilon}(R_\epsilon) \equiv-\frac{1}{n}\mathbb{E}\ln {\cal Z}_{n,\epsilon}(t; \bX, \bZ, \bZt)$ and 
$F_{n, \epsilon}(R_\epsilon) \equiv-\frac{1}{n}\ln {\cal Z}_{n,\epsilon}(t; \bX, \bZ, \bZt)$.
We have the following identities: for any given realisation of the quenched disorder
\begin{align}
 \frac{dF_{n, \epsilon}}{dR_\epsilon}  &= \langle \mathcal{L} \rangle_{t,\epsilon} \,,\label{first-derivative}\\
 \frac{1}{n}\frac{d^2F_{n, \epsilon}}{dR_\epsilon^2}  &= -\big\langle (\mathcal{L}  - \langle \mathcal{L} \rangle_{t,\epsilon})^2\big\rangle_{t,\epsilon}+
 \frac{1}{4 n^2R_\epsilon^{3/2}}\sum_{i=1}^n  \langle x_i\rangle_{t,\epsilon} \widetilde Z_i\,\label{second-derivative}.
\end{align}
Averaging \eqref{first-derivative} and \eqref{second-derivative}, using a Gaussian integration by parts w.r.t. $\widetilde Z_i$ and the Nishimori identity
$\mathbb{E}\langle x_i X_i\rangle_{t,\epsilon}  =  \mathbb{E}[\langle x_i\rangle_{t,\epsilon}^2]$ we find 
\begin{align}
 \frac{df_{n,\epsilon}}{d R_\epsilon} &= \mathbb{E}\langle \mathcal{L} \rangle_{t,\epsilon} 
 =  -\frac{1}{2n} \sum_{i=1}^n\mathbb{E}[\langle x_i\rangle_{t,\epsilon}^2]\,,\label{first-derivative-average}\\
 \frac{1}{n}\frac{d^2f_{n,\epsilon}}{dR_\epsilon^2} &= -\mathbb{E}\big\langle (\mathcal{L} - \langle \mathcal{L} \rangle_{t,\epsilon})^2\big\rangle_{t,\epsilon}
 +\frac{1}{4n^2R_\epsilon} \sum_{i=1}^n  \mathbb{E}\big\langle (x_i - \langle x_i\rangle_{t,\epsilon})^2\big\rangle_{t,\epsilon}\,.\nonumber
\end{align}
A key ingredient is again the free energy concentration, i.e., Proposition \ref{concentrationtheorem} in appendix (A2) that applies here as well (see \cite{BarbierM17a,barbier2017phase} for the proof in a inference problem, based on standards methods, which applies verbatim to the present problem). 
%
%
Consider the functions of $R_\epsilon$:
\begin{align}\label{new-free}
 & \widetilde F(R_\epsilon) \equiv F_{n, \epsilon}(R_\epsilon) +\frac{\sqrt{R_\epsilon}}{n} 
 S\sum_{i=1}^n\vert \widetilde Z_i\vert\,,
 \nonumber \\ &
 \widetilde f(R_\epsilon) \equiv \mathbb{E} \,\widetilde F(R_\epsilon)= f_{n, \epsilon}(R_\epsilon) + \frac{\sqrt{R_\epsilon}}{n} 
 S\sum_{i=1}^n \mathbb{E}\,\vert \widetilde Z_i\vert\,.
\end{align}
Because of 
\eqref{second-derivative} we see that the second derivative of $\widetilde F(R_\epsilon)$ is negative so that it is concave.
Note $F_{n,\epsilon}(R_\epsilon)$ itself is not necessarily concave in $R_\epsilon$, although $f_{n,\epsilon}(R_\epsilon)$ is.
Evidently $\widetilde f(R_\epsilon)$ is concave too.
Concavity then allows to use Lemma \ref{lemmaConvexity} in appendix (A2) as follows. First, from \eqref{new-free} we have 
\begin{align}\label{fdiff}
 \widetilde F(R_\epsilon) - \widetilde f(R_\epsilon) = F_{n, \epsilon}(R_\epsilon) - f_{n, \epsilon}(R_\epsilon) + \sqrt{R_\epsilon} S A 
\end{align} 
with $A \equiv \frac{1}{n}\sum_{i=1}^n (\vert \widetilde Z_i\vert -\mathbb{E}\,\vert \widetilde Z_i\vert)$.
Second, from \eqref{first-derivative}, \eqref{first-derivative-average} we obtain for the $R_\epsilon$-derivatives
\begin{align}\label{derdiff}
 \widetilde F'(R_\epsilon) - \widetilde f'(R_\epsilon) = 
\langle \mathcal{L} \rangle_{t,\epsilon}-\mathbb{E}\langle \mathcal{L} \rangle_{t, \epsilon} + \frac{SA}{2\sqrt{R_\epsilon}} \,.
\end{align}
From \eqref{fdiff} and \eqref{derdiff} it is then easy to show that Lemma \ref{lemmaConvexity} implies
\begin{align}\label{usable-inequ}
\vert \langle \mathcal{L}\rangle_{t,\epsilon} - \mathbb{E}\langle \mathcal{L}\rangle_{t,\epsilon}\vert&\leq 
\delta^{-1} \sum_{u\in \{R_\epsilon -\delta,\, R_\epsilon,\, R_\epsilon+\delta\}}
 \big(\vert F_{n, \epsilon}(u) - f_{n, \epsilon}(u) \vert + S\vert A \vert \sqrt{u} \big)\nonumber\\
 &\qquad\qquad\qquad\qquad
  + C_\delta^+(R_\epsilon) + C_\delta^-(R_\epsilon) + \frac{S\vert A\vert}{2\sqrt \epsilon} 
\end{align}
where $C_\delta^-(R_\epsilon)\equiv \widetilde f'(R_\epsilon-\delta)-\widetilde f'(R_\epsilon)\ge 0$ 
and $C_\delta^+(R_\epsilon)\equiv \widetilde f'(R_\epsilon)-\widetilde f'(R_\epsilon+\delta)\ge 0$. 
We used $R_\epsilon\ge \epsilon$ for the term $S\vert A\vert/(2\sqrt \epsilon)$. Note that $\delta$ will 
be chosen later on strictly smaller than $s_n$ so that $R_\epsilon -\delta \geq \epsilon - \delta \geq s_n -\delta$ remains 
positive. Remark that by independence of the noise variables $\mathbb{E}[A^2]  \le  a/n$ for some constant $a >0$. 
We square the identity \eqref{usable-inequ} and take its expectation. Then using $(\sum_{i=1}^pv_i)^2 \le p\sum_{i=1}^pv_i^2$, and 
that $R_\epsilon\le 1+\rho$, as well as the free energy concentration Proposition \ref{concentrationtheorem} of appendix (A2),
\begin{align}\label{intermediate}
 \frac{1}{9}\mathbb{E}\big[(\langle \mathcal{L}\rangle_{t,\epsilon} - \mathbb{E}\langle \mathcal{L}\rangle_{t,\epsilon})^2\big]
 &
 \leq \, 
 \frac{3}{n\delta^2} \big(C(S) +aS^2(1+\rho+\delta)\big) 
 \nonumber \\ & \qquad\qquad + C_\delta^+(R_\epsilon)^2 + C_\delta^-(R_\epsilon)^2
 + \frac{S^2a}{4n\epsilon} \, .
\end{align}
Recall $|C_\delta^\pm(R_\epsilon)|=|\widetilde f'(R_\epsilon\pm\delta)-\widetilde f'(R_\epsilon)|$. We have
\begin{align}
|\widetilde f'(R_\epsilon)|  \leq \frac12\Big(\rho  +\frac{S}{\sqrt R_\epsilon} \Big)\leq \frac12\Big(\rho  +\frac{S}{\sqrt \epsilon} \Big)\label{boudfprime}	
\end{align}
from \eqref{first-derivative-average}, \eqref{new-free} and $R_\epsilon\ge \epsilon$. Thus $|C_\delta^\pm(R_\epsilon)|\le \rho  +\frac{S}{\sqrt{\epsilon-\delta}}\le \rho  +\frac{S}{\sqrt{s_n-\delta}}$ as $\epsilon\ge s_n$. We reach
\begin{align*}
 \int_{s_n}^{2s_n} d\epsilon\, \big\{C_\delta^+(R_{\epsilon})^2 + C_\delta^-(R_{\epsilon})^2\big\}
 &\leq 
 \Big(\rho +\frac{S}{\sqrt {s_n-\delta}}\Big)
 \int_{s_n}^{2s_n} d\epsilon\, \big\{C_\delta^+(R_{\epsilon}) + C_\delta^-(R_{\epsilon})\big\}
 \nonumber \\ 
  &\leq  \Big(\rho +\frac{S}{\sqrt {s_n-\delta}}\Big)
 \int_{R_{s_n}}^{R_{2s_n}} dR_\epsilon\, \big\{C_\delta^+(R_\epsilon) + C_\delta^-(R_\epsilon)\big\}
 \nonumber \\ 
&\hspace{-4cm}= 
\Big(\rho +\frac{S}{\sqrt {s_n-\delta}}\Big)\Big[\Big(\widetilde f_{n,\epsilon}(R_{s_n}+\delta) - \widetilde f_{n,\epsilon}(R_{s_n}-\delta)\Big)+ \Big(\widetilde f_{n,\epsilon}(R_{2s_n}-\delta) - \widetilde f_{n,\epsilon}(R_{2s_n}+\delta)\Big)\Big]
\end{align*}
where we used that the Jacobian of the ${\cal C}^1$ diffeomorphism $\epsilon\mapsto R_\epsilon(t)$ is $\ge 1$ (by regularity) for 
the second inequality. The mean value theorem and \eqref{boudfprime} 
imply $|\widetilde f_{n, \epsilon}(R_\epsilon-\delta) - \widetilde f_{n, \epsilon}(R_{\epsilon}+\delta)|\le \delta(\rho  +\frac{S}{\sqrt{s_n-\delta}})$. Therefore
\begin{align*}
 \int_{s_n}^{2s_n} d\epsilon\, \big\{C_\delta^+(R_{\epsilon})^2 + C_\delta^-(R_{\epsilon})^2\big\}\leq 
 2\delta \Big(\rho +\frac{S}{\sqrt{s_n-\delta}}\Big)^2\,.
\end{align*}
Thus, integrating \eqref{intermediate} over $\epsilon\in [s_n, 2s_n]$ yields
\begin{align*}
 &\int_{s_n}^{2s_n} d\epsilon\, 
 \mathbb{E}\big[(\langle \mathcal{L}\rangle_{t,\epsilon} - \mathbb{E}\langle \mathcal{L}\rangle_{t, \epsilon})^2\big]\nonumber\\
 &\qquad\qquad\leq 27\big(C(S) +aS^2(1+\rho+\delta)\big) \frac{s_n}{n\delta^2}  +18\delta \Big(\rho +\frac{S}{\sqrt {s_n-\delta}}\Big)^2 +  \frac{9 S^2a\ln 2}{4n}  \,.
\end{align*}
Finally we optimize the bound choosing $\delta  =  s_n^{2/3} n^{-1/3}$ and obtain the desired result.	
\end{proof}

\subsection*{(A6) Proof of inequality \eqref{remarkable}}\label{proof:remarkable_id}
Let us drop the indices in the bracket $\langle -\rangle_{t,\epsilon}$ and recall $R_\epsilon\equiv R_\epsilon(t)$. We start by proving 
\begin{align}
-2\,\mathbb{E}\big\langle Q(\mathcal{L} - \mathbb{E}\langle \mathcal{L}\rangle)\big\rangle
&=\mathbb{E}\big\langle (Q - \mathbb{E}\langle Q \rangle)^2\big\rangle
+ \mathbb{E}\big\langle (Q-   \langle Q \rangle)^2\big\rangle\,.\label{47}
\end{align}
Using the definitions $Q \equiv \frac1n \sum_{i=1}^{n} x_i X_i$ and \eqref{def_L} gives
\begin{align}
2\,\mathbb{E}\big\langle Q ({\cal L} -\mathbb{E}\langle {\cal L} \rangle) \big\rangle
	& = \frac{1}{n^2} \sum_{i,j=1}^{n} \Big\{\mathbb{E} \Big [  X_i \langle x_i x_j^2 \rangle - 2X_i X_j \langle x_i x_j \rangle - \frac{\widetilde Z_j}{\sqrt{R_\epsilon}} X_i \langle x_i x_j \rangle \Big ] \nonumber \\
	& \qquad\qquad \quad-  \mathbb{E} [X_i \langle x_i \rangle ] \, \mathbb{E} \Big [ \langle x_j^2 \rangle - 2X_j \langle x_j \rangle - \frac{\widetilde Z_j}{\sqrt{R_\epsilon}} \langle x_j \rangle \Big ]\Big\}\,. \label{eq:QL:1}
\end{align}
The Gaussian integration by part formula $\mathbb{E}[Zg(Z)]=\mathbb{E}\,g'(Z)$ for $Z\sim {\cal N}(0,1)$ yields
\begin{align*}
\mathbb{E}\Big[\frac{\widetilde Z_j}{\sqrt{R_\epsilon}} X_i \langle x_i x_j \rangle\Big] 
	& = \mathbb{E}[ X_i \langle x_i x_j^2 \rangle - X_i\langle x_i x_j \rangle \langle x_j \rangle ]\,, \quad\text{and} \quad \mathbb{E}\Big[\frac{\widetilde Z_j}{\sqrt{R_\epsilon}} \langle x_j \rangle\Big] 
	 =   \mathbb{E}[ \langle x_j^2 \rangle - \langle x_j \rangle^2 ]\,.
\end{align*}
These simplify \eqref{eq:QL:1} to
\begin{align}
&2\,\mathbb{E}\big\langle Q ({\cal L} -\mathbb{E}\langle {\cal L} \rangle ) \big\rangle\nonumber\\
	 &\qquad\qquad= \frac{1}{n^2} \sum_{i,j=1}^{n}\big\{ \mathbb{E} [ X_i \langle x_j \rangle \langle x_i x_j \rangle - 2X_i X_j \langle x_i x_j \rangle ] -  \mathbb{E} [X_i \langle x_i \rangle ] \,\mathbb{E} [ \langle x_j \rangle^2 - 2X_j \langle x_j \rangle ]\big\}\,.\label{toSimp}
\end{align}
The Nishimori identity Lemma~\ref{NishId} implies
\begin{align*}
 \mathbb{E}[\langle x_j \rangle^2 ] = \mathbb{E}[X_j \langle x_j \rangle]\,, \quad \text{and}\quad \mathbb{E}[X_i \langle x_j \rangle \langle x_i x_j \rangle] = \mathbb{E}[ \langle x_i\rangle \langle x_j \rangle \langle x_i x_j \rangle ] = \mathbb{E}[ \langle x_i \rangle \langle x_j \rangle X_i X_j ]\,.
\end{align*}
These further simplify \eqref{toSimp} to
\begin{align*}
2\,\mathbb{E}\big\langle Q ({\cal L} -\mathbb{E}\langle {\cal L} \rangle) \big\rangle & = \frac{1}{n^2} \sum_{i,j=1}^{n}\big\{ \mathbb{E} [  \langle x_i \rangle \langle x_j \rangle X_i X_j - 2X_i X_j \langle x_i x_j \rangle ] +  \mathbb{E} [X_i \langle x_i \rangle ] \,\mathbb{E} [ X_j \langle x_j \rangle ]\big\}\nonumber\\
	& = \mathbb{E}[\langle Q\rangle^2] - 2\,\mathbb{E}\langle Q^2\rangle +  \mathbb{E}[\langle Q\rangle]^2 \\
	& = -  \big ( \mathbb{E}\langle Q^2\rangle - \mathbb{E}[\langle Q\rangle]^2 \big ) -  \big ( \mathbb{E}\langle Q^2\rangle - \mathbb{E}[\langle Q\rangle^2] \big ) 
\end{align*}
which is \eqref{47}. 

The identity \eqref{47} just proven then implies
\begin{align*}
2\big|\mathbb{E}\big\langle Q(\mathcal{L} - \mathbb{E}\langle \mathcal{L}\rangle)\big\rangle\big|=2\big|\mathbb{E}\big\langle (Q-\mathbb{E}\langle Q \rangle)(\mathcal{L} - \mathbb{E}\langle \mathcal{L}\rangle)\big\rangle\big|
\ge \mathbb{E}\big\langle (Q - \mathbb{E}\langle Q \rangle)^2\big\rangle
\,.
\end{align*}
An application of the Cauchy-Schwarz inequality then gives
\begin{align*}
2\big\{\mathbb{E}\big\langle (Q-\mathbb{E}\langle Q \rangle)^2\big\rangle\, \mathbb{E}\big\langle(\mathcal{L} - \mathbb{E}\langle \mathcal{L}\rangle)^2\big\rangle \big\}^{1/2}	\ge\mathbb{E}\big\langle (Q - \mathbb{E}\langle Q \rangle)^2\big\rangle\,.
\end{align*}
This ends the proof of \eqref{remarkable}.

\bibliographystyle{unsrt_abbvr}      
\bibliography{refs}   


\end{document}